\documentclass[preprint,showpacs,nofootinbib,floatfix]{revtex4-1} 
\usepackage{graphicx}
\usepackage{amssymb}
\usepackage{amsmath}
\usepackage{amsthm}
\usepackage{amsfonts}
\usepackage{mathrsfs}
\usepackage{color}
\usepackage{hyperref}
\usepackage{mathtools}

\newcommand{\col}[1]{{\mathrm{color}(#1)}}
\renewcommand{\star}[2]{{\mathrm{St}_#1(#2)}}
\newcommand{\face}[2]{{\Delta_#1(#2)}}
\newcommand{\link}[2]{{\mathrm{Lk}_#1(#2)}}

\newcommand{\bpm}{\begin{pmatrix}}
\newcommand{\epm}{\end{pmatrix}}

\theoremstyle{lemma}
\newtheorem{lemma}{Lemma}

\theoremstyle{theorem}
\newtheorem{theorem}{Theorem}

\theoremstyle{definition}

\theoremstyle{remark}

\theoremstyle{conjecture}

\theoremstyle{observation}

\theoremstyle{definition}
\theoremstyle{corollary}

\theoremstyle{result}
\newtheorem{result}{Result}

\theoremstyle{cond}

\newtheorem{example}{Example}

\begin{document}

\title{Unfolding the color code}

\author{Aleksander Kubica}
\author{Beni Yoshida}
\author{Fernando Pastawski}
\affiliation{Institute for Quantum Information \& Matter and Walter Burke Institute for Theoretical Physics, California Institute of Technology,  Pasadena CA 91125, USA}

\date{\today}

\begin{abstract}
The topological color code and the toric code are two leading candidates for realizing fault-tolerant quantum computation. Here we show that the color code on a $d$-dimensional closed manifold is equivalent to multiple decoupled copies of the $d$-dimensional toric code up to local unitary transformations and adding or removing ancilla qubits. Our result not only generalizes the proven equivalence for $d=2$, but also provides an explicit recipe of how to decouple independent components of the color code, highlighting the importance of colorability in the construction of the code. Moreover, for the $d$-dimensional color code with $d+1$ boundaries of $d+1$ distinct colors, we find that the code is equivalent to multiple copies of the $d$-dimensional toric code which are attached along a $(d-1)$-dimensional boundary. In particular, for $d=2$, we show that the (triangular) color code with boundaries is equivalent to the (folded) toric code with boundaries. We also find that the $d$-dimensional toric code admits logical non-Pauli gates from the $d$-th level of the Clifford hierarchy, and thus saturates the bound by Bravyi and K\"{o}nig. In particular, we show that the $d$-qubit control-$Z$ logical gate can be fault-tolerantly implemented on the stack of $d$ copies of the toric code by a local unitary transformation.
\end{abstract}

\maketitle

\tableofcontents
\vfill\eject

\section{Introduction}

Quantum error-correcting codes~\cite{Shor96,Preskill98} are vital for fault-tolerant realization of quantum information processing tasks. Of particular importance are topological quantum codes~\cite{Kitaev03, Dennis02} where quantum information is stored in non-local degrees of freedom while the codes are characterized by geometrically local generators. An essential feature of such codes is to admit a fault-tolerant implementation of a universal gate set as this would guarantee that the physical errors propagate in a benign and controlled manner. Thus, the search for novel quantum error-correcting codes and the classification of fault-tolerantly implementable logical gates in these codes have been central problems in quantum information science~\cite{Eastin09, Bravyi13, Paetznick13, Pastawski15, Beverland14}.

The quest of analyzing topological quantum codes is also closely related to the central problem in quantum many-body physics, namely the classification of quantum phases~\cite{Sachdev_Text,Chen10}. A fruitful approach is to view topological quantum codes as exactly solvable toy models which correspond to representatives of gapped quantum phases. This approach has led to a complete classification of translation symmetric two-dimensional stabilizer Hamiltonians~\cite{Beni10b,Bombin14b}, as well as to the discovery of a novel three-dimensional topological phase which does not fit into previously known theoretical framework~\cite{Haah11,Beni13}.

Topological color codes~\cite{Bombin06} are important examples of topological stabilizer codes that admit transversal implementation of a variety of logical gates, which may not be fault-tolerantly implementable in other topological stabilizer codes. 
In two spatial dimensions, the color code admits transversal implementation of all the Clifford logical gates. 
In three and higher dimensions, the color code admits transversal implementation of non-Clifford logical gates~\cite{Bombin09}. 
A naturally arising question is to identify the physical properties allowing to extend the set of transversally implementable logical gates with respect to other topological codes.

Given two codes with different sets of fault-tolerantly implementable logical gates, one may naturally expect that they correspond to different topological phases of matter. However, physical properties of color codes and toric codes are known to be very similar. 
For instance, both of the codes have logical Pauli operators with similar geometric shapes, which leads to essentially identical braiding properties of anyonic excitations from the viewpoint of long-range physics. Furthermore, it has been proven that translation symmetric stabilizer codes, supported on a two-dimensional torus, are equivalent to multiple decoupled copies of the two-dimensional toric code up to local unitary transformations and adding or removing ancilla qubits~\cite{Beni10b}. This result implies that the two-dimensional color code supported on a torus is equivalent to two decoupled copies of the toric code, and thus they belong to the same quantum phase~\cite{Bombin11}. 

However, the aforementioned results do not consider the effect of boundaries on the classification of quantum phases~\cite{Bravyi98,Kitaev12,Beigi11}. In fact, the color code admits transversal implementation of computationally useful logical gates only if it is supported on a system with appropriately designed boundaries. Perhaps, the presence of boundaries may render additional computational power to topological quantum codes and may result in richer structure of topological phases of matter. Complete understanding of the relation between the color code and the toric code will be the necessary first step to clarify the connection between boundaries and achievable fault-tolerant logical gates, and its implication to the classification of quantum phases.

\subsection{Summary of main results}

In this paper, we establish a connection between the color code and the toric code in the presence or absence of boundaries, and study fault-tolerantly implementable logical gates in these two codes. Our first result, presented in Section~\ref{sec:closed}, focuses on the equivalence between the color code and the toric code on $d$-dimensional lattices without boundaries, $d\geq 2$.

\begin{result}[Closed manifold]
\emph{
The topological color code on a $d$-dimensional closed manifold (without boundaries) is equivalent to multiple decoupled copies of the $d$-dimensional toric code up to local unitary transformations and adding or removing ancilla qubits.}
\end{result}

This extends the known results from~\cite{Beni10b,Bombin14b,Bombin11} to the family of color codes in arbitrary dimensions. While previous results are limited to either translation symmetric systems or do not provide an explicit method of transformations, we provide an specific construction of how to decouple the color code defined on an arbitrary $d$-dimensional manifold into multiple decoupled toric code components. The recipe emphasizes the importance of colorability in the construction of the color code. Our result implies that the topological color code and the toric code belong to the same quantum phase according to the definition widely accepted in condensed matter physics community~\cite{Chen10}.

In Section~\ref{sec:open}, we analyze the $d$-dimensional topological color code with boundaries. The second result concerns the equivalence for systems with boundaries.

\begin{result}[Boundaries]
\emph{The $d$-dimensional topological color code on with boundaries is equivalent to $d$ copies of the $d$-dimensional toric code which are attached along a $(d-1)$-dimensional boundary.}
\end{result}

\vspace*{-5pt}
\begin{figure}[h!]
\includegraphics[width=0.70\textwidth]{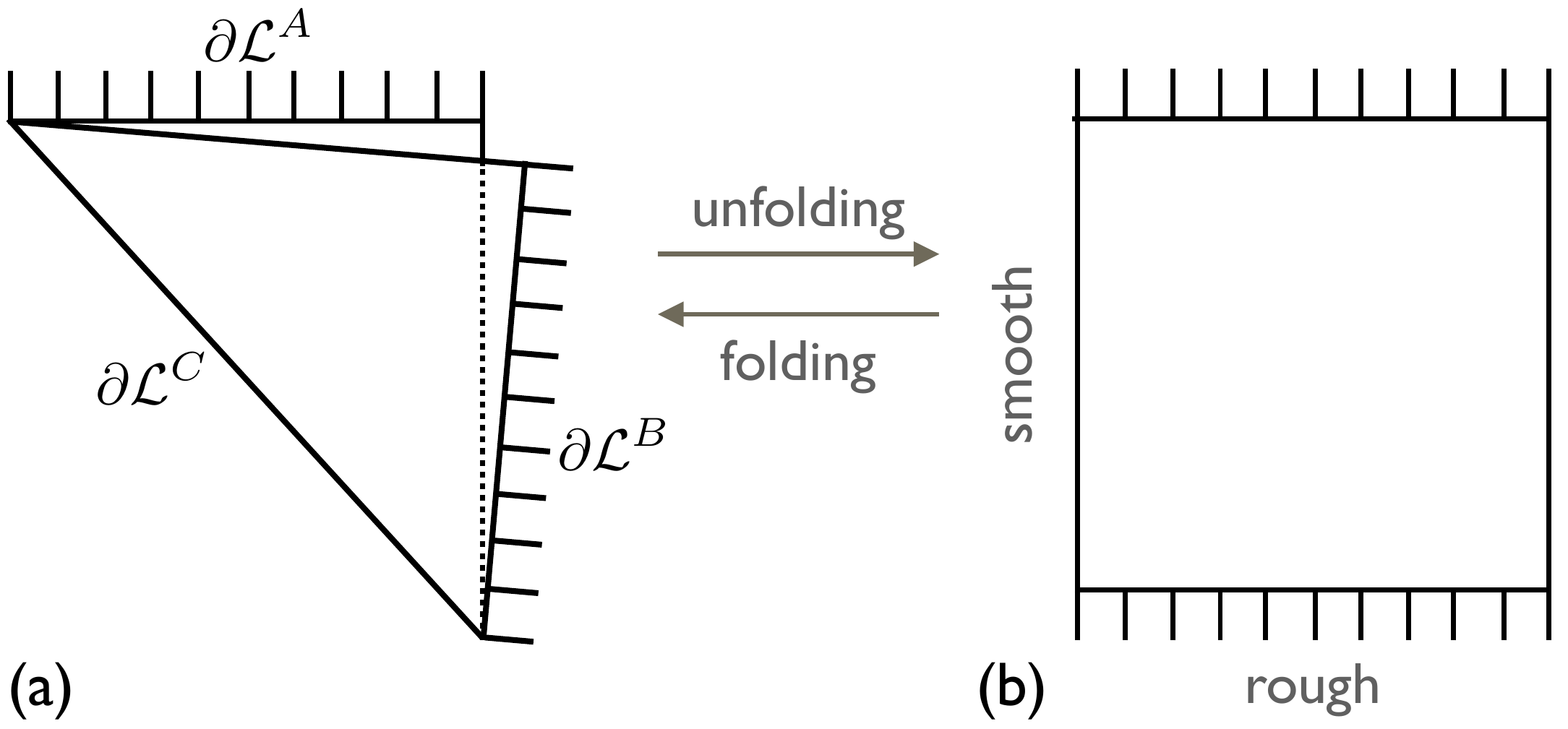}
\vspace*{-5pt}
\caption{The topological color code (a) with three boundaries $\partial\mathcal{L}^A$, $\partial\mathcal{L}^B$ and $\partial\mathcal{L}^C$ viewed as the folded toric code (b) with two smooth and two rough boundaries. The boundary $\partial\mathcal{L}^A$ of color $A$ is equivalent to a pair of boundaries --- smooth in the front and rough in the rear layer; similarly $\partial\mathcal{L}^B$. The boundary $\partial\mathcal{L}^C$ is the fold.}
\label{fig_folding} 
\end{figure} 

In two dimensions, we find that the (triangular) color code with three boundaries is equivalent to the toric code with boundaries (i.e. the surface code) which is folded (see Fig.~\ref{fig_folding}). For $d>2$, we find that the color code with $d+1$ boundaries of $d+1$ distinct colors is equivalent to $d$ copies of the toric code which are attached along a $(d-1)$-dimensional boundary. On this $(d-1)$-dimensional boundary, a composite electric charge composed of all $d$ electric charges from the different copies of the toric codes may condense. Other boundaries are decoupled and allow condensation of a single electric charge associated with a specific copy.

In Section~\ref{sec:gate}, we study non-Clifford logical gates fault-tolerantly implementable in the $d$-dimensional toric code. Our third result concerns the implementability of the $d$-qubit control-$Z$ gate, i.e. a gate which applies $-1$ phase only if all $d$ qubits are in $|1\rangle$ state. 

\begin{result}[Logical gate]
\emph{A stack of $d$ copies of the $d$-dimensional toric code with point-like logical excitations admits fault-tolerant implementation of the logical $d$-qubit control-$Z$ gate by local unitary transformations.}
\end{result}

In particular, we find that transversal application of physical $R_{d}=\mathrm{diag}(1,e^{2\pi i/2^d})$ phase gates on the $d$-dimensional topological color code is equivalent to the logical \mbox{$d$-qubit} \mbox{control-$Z$} gate acting on $d$ copies of the toric code. Note that the $d$-qubit \mbox{control-$Z$} gate belongs to the $d$-th level of the Clifford hierarchy, but is outside of the $(d-1)$-th level. Thus, a stack of $d$ copies of the $d$-dimensional toric code saturates the bound by Bravyi and K\"{o}nig on fault-tolerant logical gates which are implementable by local unitary transformations~\cite{Bravyi13b}. For a definition of the Clifford hierarchy, see~\cite{Gottesman99, Bravyi13b, Pastawski15}\\ 

We believe that our findings will shed light on the techniques of code deformation~\cite{Bombin09} and lattice surgery~\cite{Horsman12,Landahl14}, allowing for computation with less physical qubits, higher fault-tolerant error suppression and shorter time. The ability to transform and relate different codes may turn out to be crucial in analyzing the available methods of computation with topological codes. In particular, we might be able to improve the decoding scheme for the color code proposed in Ref.~\cite{Delfosse2014}, and generalize it to any dimensions. Also, our findings may lead to a systematic method of composing known quantum codes to construct new codes with larger set of fault-tolerant logical gates. Finally, an interesting future problem is to apply the disentangling unitary to the gauge color codes~\cite{Bombin2013,Kubica15}.

While the paper is written in a relatively self-contained manner, we assume some prior exposure to the construction of the topological color code. A pedagogical description of the color code has been given by one of the authors in Ref.~\cite{Kubica15}. Our discussion mostly concerns the $d$-dimensional topological color code and the toric code with point-like excitations as it is the most interesting case from the viewpoint of transversal non-Clifford gates. For the sake of simplicity, we present proof sketches relying on many figures. Rigorous proofs require the language of algebraic topology~\cite{Glaser_Text,Hatcher2002}, which might be technically challenging and could obscure the main ideas presented in the paper. Thus, we postpone them until the Appendix.

\section{Topological color code without boundaries}
\label{sec:closed}

In this section, we show that the $d$-dimensional topological color code supported on a closed manifold is equivalent to multiple decoupled copies of the toric code. 

\subsection{Brief introduction to the color code and the toric code}

We begin by briefly reviewing the construction of the topological color code and the toric code. The starting point to define either the toric code or the color code is a two-dimensional lattice $\mathcal{L}$. We can think of $\mathcal{L}$ as a homogeneous cell $2$-complex, i.e. a collection of vertices $V$, edges $E$ and faces $F$, glued together in a certain way. In general, $\mathcal{L}$ can be defined on a manifold with boundaries, but in this section we restrict our attention to closed manifolds.

The toric code in two dimensions is defined on a lattice $\mathcal{L}$ by placing one qubit on every edge, and associating $X$- and $Z$-type stabilizer generators with vertices and faces of $\mathcal{L}$, namely
\begin{eqnarray}
\forall v\in V: X(v)=\bigotimes_{e\supset v} X(e),\qquad
\forall f\in F: Z(f)=\bigotimes_{e\subset f} Z(e).
\end{eqnarray}
Here, $X(e)$ and $Z(e)$ denote Pauli $X$ and $Z$ operators on the qubit placed on the edge $e$; see Fig.~\ref{fig_toric_color}(a) for an example. We denote such a code, as well as its stabilizer group by $TC(\mathcal{L})$. One can verify that $X$- and $Z$-type stabilizer generators commute.

The color code is defined on a lattice $\mathcal{L}$, which satisfies two additional conditions:
\begin{itemize}
\item valence --- each vertex belongs to exactly three edges,
\item colorability --- there is a coloring\footnote{Note that due to the valence condition, this coloring is unique up to permutation of colors for any connected component of the lattice $\mathcal{L}$.} of faces of $\mathcal{L}$ with three colors, $A$, $B$ and $C$, such that any two adjacent faces have different colors.
\end{itemize}
For instance, the honeycomb lattice satisfies the valence and colorability conditions; also see Fig.~\ref{fig_toric_color}(b). In the case of the color code, we place one qubit at every vertex, and associate $X$- and $Z$-type stabilizer generators with every face of $\mathcal{L}$, namely
\begin{eqnarray}
\forall f\in F: X(f)=\bigotimes_{v\subset f} X(v),\qquad
\forall f\in F: Z(f)=\bigotimes_{v\subset f} Z(v).
\end{eqnarray}
To verify that $X$- and $Z$-type stabilizers commute, one uses the valence and colorability conditions. We denote such a code, as well as its stabilizer group by $CC(\mathcal{L})$. 

We can generalize the definition of the toric code and the color code to $d$ dimensions by considering a $d$-dimensional lattice (i.e. a homogeneous cell $d$-complex) $\mathcal{L}$. There are $d-1$ different ways of defining the toric code on $\mathcal{L}$ --- place qubits on $m$-cells, $m=1,2,\ldots,d-1$, and associate $X$- and $Z$-type stabilizer generators with $(m-1)$- and $(m+1)$-cells, respectively. In the case of the color code, the additional conditions are that $\mathcal{L}$ is $(d+1)$-valent and its $d$-cells are $(d+1)$-colorable. There are $d-1$ ways of defining the color code on $\mathcal{L}$ --- place qubits on vertices, and associate $X$- and $Z$-type stabilizer generators with $m$- and $(d+2-m)$-cells, where $m=2,3,\ldots,d$. For a rigorous definition of the toric code and the color code in $d$ dimensions see the Appendix.

In the main body of the paper, we restrict our attention to the color code and toric code with point-like excitations, which significantly simplifies the discussion. In particular, the color code has $X$- and $Z$-type stabilizers associated with $d$-cells and $2$-cells (faces), whereas the toric code has qubits placed on  edges. We postpone the discussion of the general case until the Appendix.
 
\begin{figure}[h!]
\includegraphics[width=0.65\textwidth]{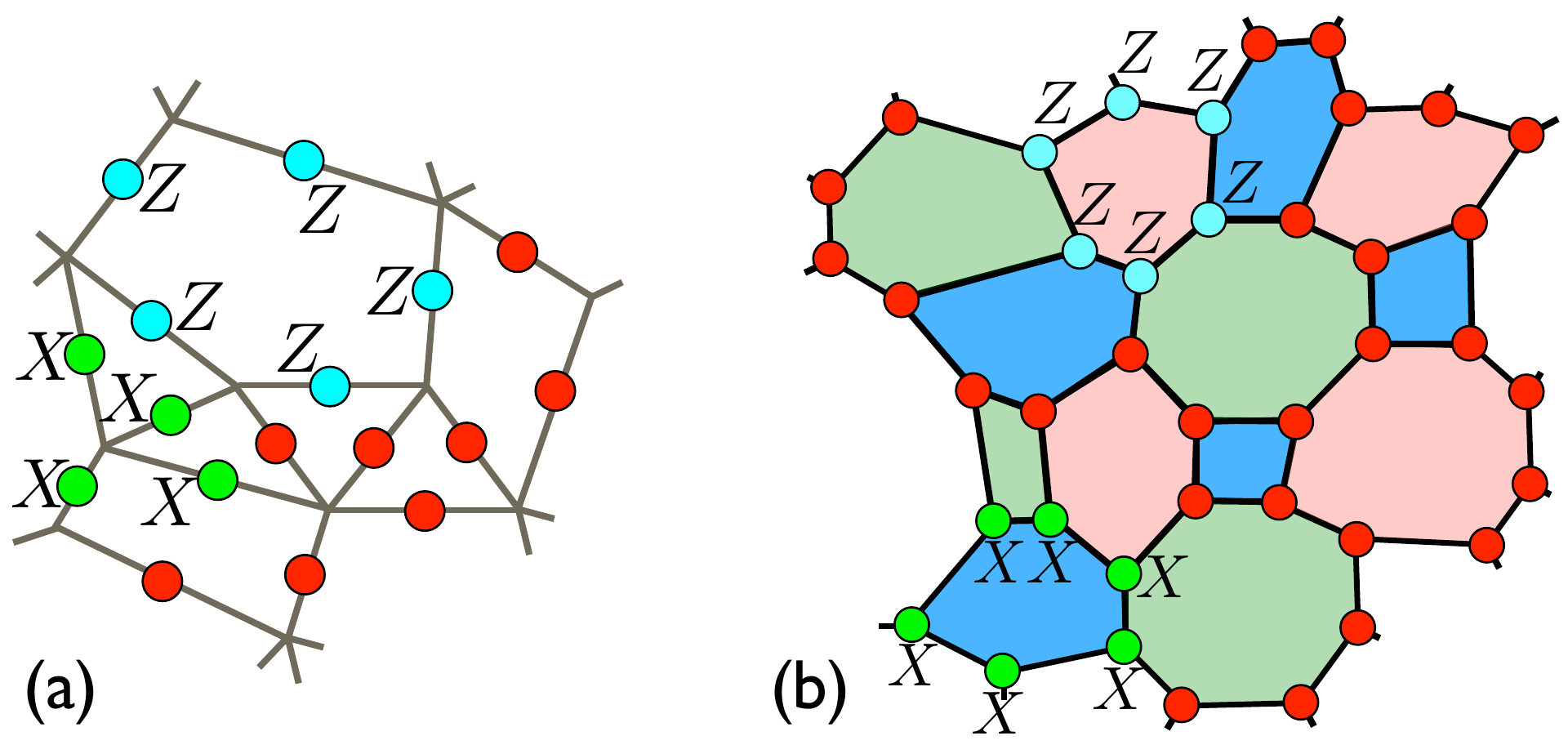}
\caption{(Color online) The toric code and the color code in two dimensions. (a) The toric code has qubits (red dots) placed on edges, and $X$-vertex (green) and $Z$-face (blue) stabilizer generators. (b) The color code has qubits placed on vertices, and $X$-face and $Z$-face stabilizer generators. Note that the color code can only be defined on a $3$-valent and $3$-colorable lattice.}
\label{fig_toric_color} 
\end{figure}

\subsection{Equivalence in two dimensions}

In this subsection, we prove that the two-dimensional color code supported on a closed manifold (without boundaries) is equivalent to two copies of the toric code. 

\begin{theorem}
\label{th:equivin2d}
Let  $CC(\mathcal{L})$  be the two-dimensional topological color code defined on a lattice $\mathcal{L}$ without boundaries, colored in $A$, $B$ and $C$. There exists a local Clifford unitary $U$, and two lattices $\mathcal{L}_{A}$ and $\mathcal{L}_{B}$ obtained from $\mathcal{L}$ by shrinking faces of color $A$ and $B$, respectively, such that
\begin{equation}
U[CC(\mathcal{L})] U^{\dagger}= TC(\mathcal{L}_{A})\otimes TC(\mathcal{L}_{B}).
\end{equation}
Moreover, one can choose $U$ to be
\begin{equation}
U= \bigotimes_{f \in \mathcal{C}} U_{f},
\end{equation}
where $\mathcal{C}$ represents the set of all faces in $\mathcal{L}$ colored with $C$, and $U_{f}$ is a Clifford unitary acting only on qubits of the face $f$. 
\end{theorem}

Here, the tensor product $TC(\mathcal{L}_{A})\otimes TC(\mathcal{L}_{B})$ indicates that the stabilizer group can be factored into two independent stabilizer groups associated with two decoupled copies of the toric code on the lattices $\mathcal{L}_{A}$ and $\mathcal{L}_{B}$. We shall refer to $\mathcal{L}_{A}$ and $\mathcal{L}_{B}$ supporting two decoupled copies of the toric code as \emph{shrunk lattices} (see~Refs.~\cite{Bombin2007,Bombinbook}).

As described in the theorem, the disentangling unitary transformation $U$ has a tensor product structure, $U= \bigotimes_{f \in \mathcal{C}} U_{f}$. Thus, $U$ is a local unitary transformation, and two systems belong to the same quantum phase. 

The procedure of decoupling two copies of the toric code starting from the color code consists of three steps:
\begin{enumerate}
\item performing certain local unitary $U_f$ at each and every face $f$ of color $C$ in $\mathcal{L}$,
\item checking that the stabilizer generators $CC(\mathcal{L})$ are mapped by $U=\bigotimes_{f\in \mathcal{C}} U_f$ into two set of generators $TC(\mathcal{L}_{A})$ and  $TC(\mathcal{L}_{B})$ supported on disjoint sets of qubits,
\item visualizing two codes $TC(\mathcal{L}_{A})$ and  $TC(\mathcal{L}_{B})$ as codes defined on lattices $\mathcal{L}_{A}$ and $\mathcal{L}_{B}$ obtained from $\mathcal{L}$ by local deformations.
\end{enumerate}

\textbf{Step 1:}
Let us pick a face $f$ of $\mathcal{L}$ colored in $C$. Since $\mathcal{L}$ is $3$-colorable and $3$-valent, the face $f$ has even number  of vertices, $2n$. Let us enumerate vertices of $f$ in counter-clockwise order in such a way that the edge $(1,2)$ between vertices $1$ and $2$ has color $AC$. We would like to find a unitary transformation $U_f$ of the Hilbert space $\mathcal{H}_V$ of qubits placed on vertices into the Hilbert space $\mathcal{H}_E$ of qubits placed on edges\footnote{Note that since the number of vertices of $f$ is equal to the number of edges of $f$, then $\mathcal{H}_V\simeq\mathcal{H}_E\simeq (\mathbb{C}^2)^{\otimes 2n}$. Moreover, to perform such a transformation, one does not need any ancilla qubits.} such that some operators on $\mathcal{H}_V$ are mapped into certain operators on $\mathcal{H}_E$.  In particular, we would require the following mappings to hold
\begin{eqnarray}
\label{eq:list1}
Z_{j}Z_{j+1}\qquad &\rightarrow&\qquad Z_{(j,j+1)} \qquad\qquad\qquad\qquad\qquad (j=1,\ldots,2n-1),\\
\left( \prod_{j=1}^{2n} X_j \right)\cdot Z_{2n}Z_{1}\qquad &\rightarrow&\qquad Z_{(2n,1)},\\
X_{j}X_{j+1}\qquad &\rightarrow&\qquad X_{(j-1,j)}X_{(j+1,j+2)}  \qquad\qquad\qquad (j=1,\ldots,2n-2),\\
\left( \prod_{j=1}^{2n} X_j \right)\cdot X_{j}X_{j+1}\qquad &\rightarrow&\qquad X_{(j-1,j)}X_{(j+1,j+2)} \qquad\qquad\qquad (j=2n-1,2n),
\label{eq:list2}
\end{eqnarray}
where $X_{j}$ represents Pauli $X$ operator on a qubit on the vertex $j$, while $X_{(j,j+1)}$ represents Pauli $X$ operator on a qubit on the edge $(j,j+1)$ and $2n+1\equiv 1$; similarly for $Z_{j}$ and $Z_{(j,j+1)}$. The conditions imposed on $U_f$ by Eqs.~(\ref{eq:list1})--(\ref{eq:list2}) for the face $f$ with six vertices are illustrated in Fig.~\ref{fig_2d_transformation}.

\begin{figure}[h!]
\includegraphics[width=0.85\textwidth]{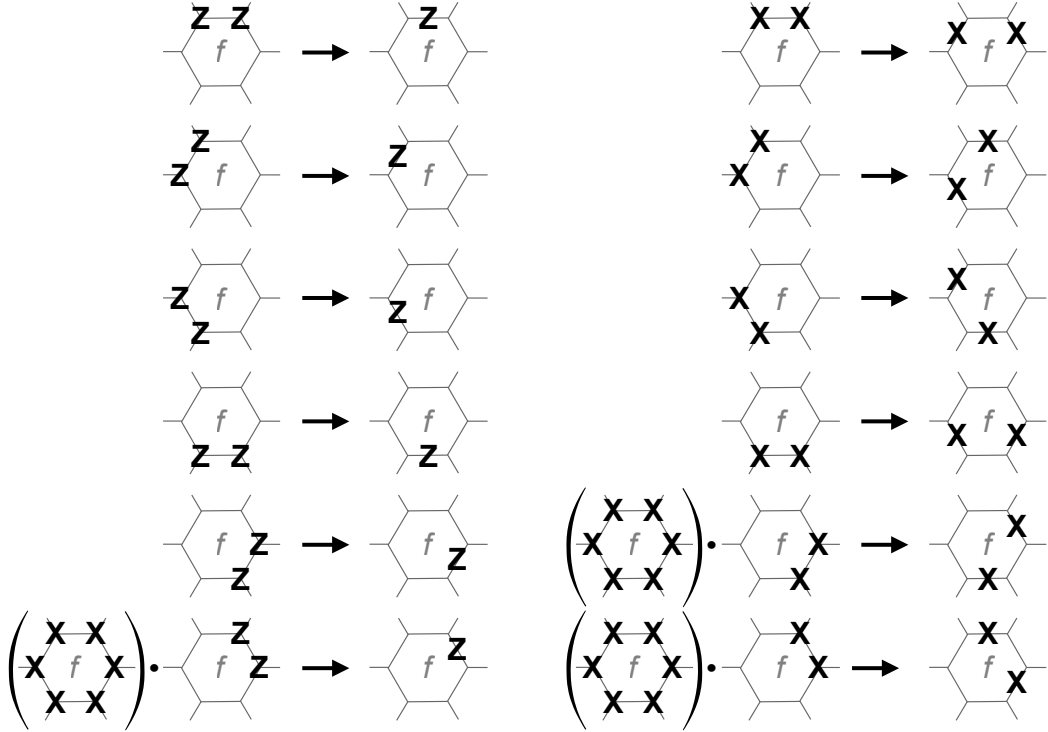}
\caption{Transformation of the operators of the color code $CC(\mathcal{L})$ supported on qubits of the face $f$ colored in $C$ under the disentangling unitary transformation $U_{f}$.}
\label{fig_2d_transformation} 
\end{figure} 

We claim that there exists a Clifford unitary $U_{f}$ which satisfies~Eqs.~(\ref{eq:list1})--(\ref{eq:list2}). The proof of existence of such a unitary transformation is presented later. Note that under the unitary $U_f$ the operators on the qubits on vertices of $f$ (up to the stabilizer $\prod_{j=1}^{2n} X_j$)  transform into the operators on the qubits placed on edges of $f$ in the following way
\begin{align}
\includegraphics[height=0.65in]{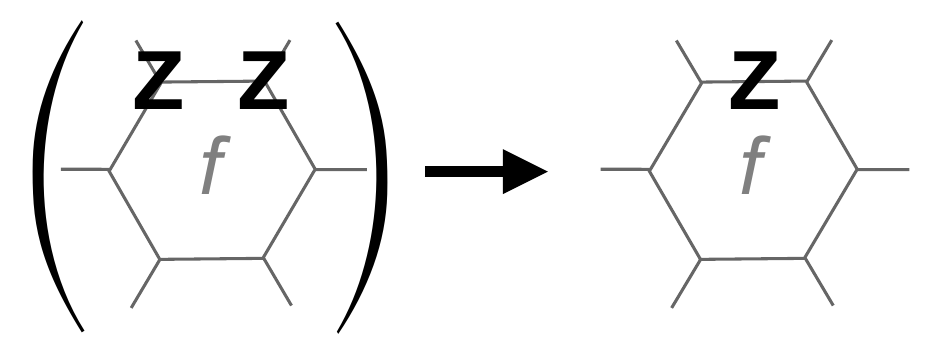} \qquad
\includegraphics[height=0.65in]{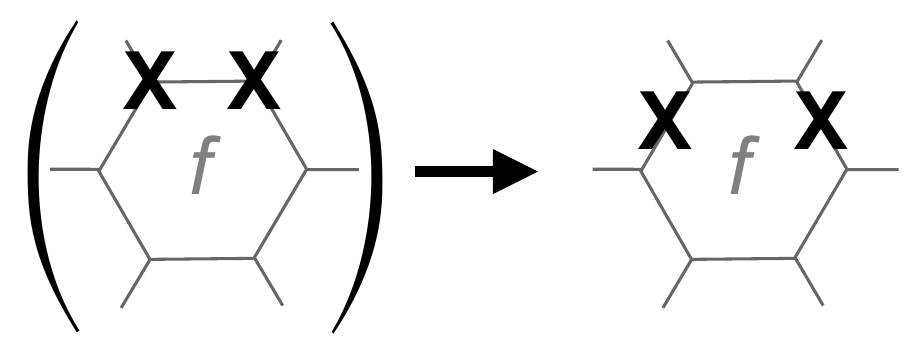}
\end{align}
where parenthesis indicate that operators might be multiplied by the stabilizer $\prod_{j=1}^{2n} X_j$.

\textbf{Step 2:}
Let us analyze what happens to the stabilizer generators $CC(\mathcal{L})$ of the color code after performing Step 1 for each and every face of color $C$, i.e. after action of $U=\bigotimes_{f\in\mathcal{C}} U_f$ by conjugation. Note that the stabilizer group $CC(\mathcal{L})$ does not have a unique representation in terms of its generators --- for instance, $CC(\mathcal{L})$ can be generated by $Y$- and $Z$-type stabilizers associated with every face of $\mathcal{L}$.

For a face $f$ colored in $C$, the unitary $U$ transforms the $Z$-face stabilizer on qubits on vertices into the $Z$-type operator on qubits on edges colored in $AC$. Similarly, the $Y$-face stabilizer on vertices is transformed into the $Z$-type operator on qubits on edges of $f$ colored in $BC$. For a face $f'$ colored in $A$ (respectively $B$), the unitary $U$ transforms the $Z$-face stabilizer on vertices (up to multiplication by $X$-face stabilizers on faces of color $C$ neighboring $f'$ --- this depends on the choice of $U$ in Step 1) into the $Z$-type operator on qubits on edges of $f'$ colored in $AC$ (respectively $BC$). On the other hand, the $X$-face stabilizer is transformed into the $X$-type operator on qubits on edges radiating out of $f'$, which are colored in $BC$ (respectively $AC$).

Fig.~\ref{fig_2d_terms} summarizes how the stabilizers of the color code transform under $U$ described in~Fig. \ref{fig_2d_transformation}. The parenthesis to the left indicate that the stabilizer of the color code might be multiplied by the $X$-face stabilizers on neighboring faces of color $C$, depending on the disentangling procedure, i.e. the choice of $U$.

\begin{figure}[h!]
\includegraphics[width=.75\textwidth]{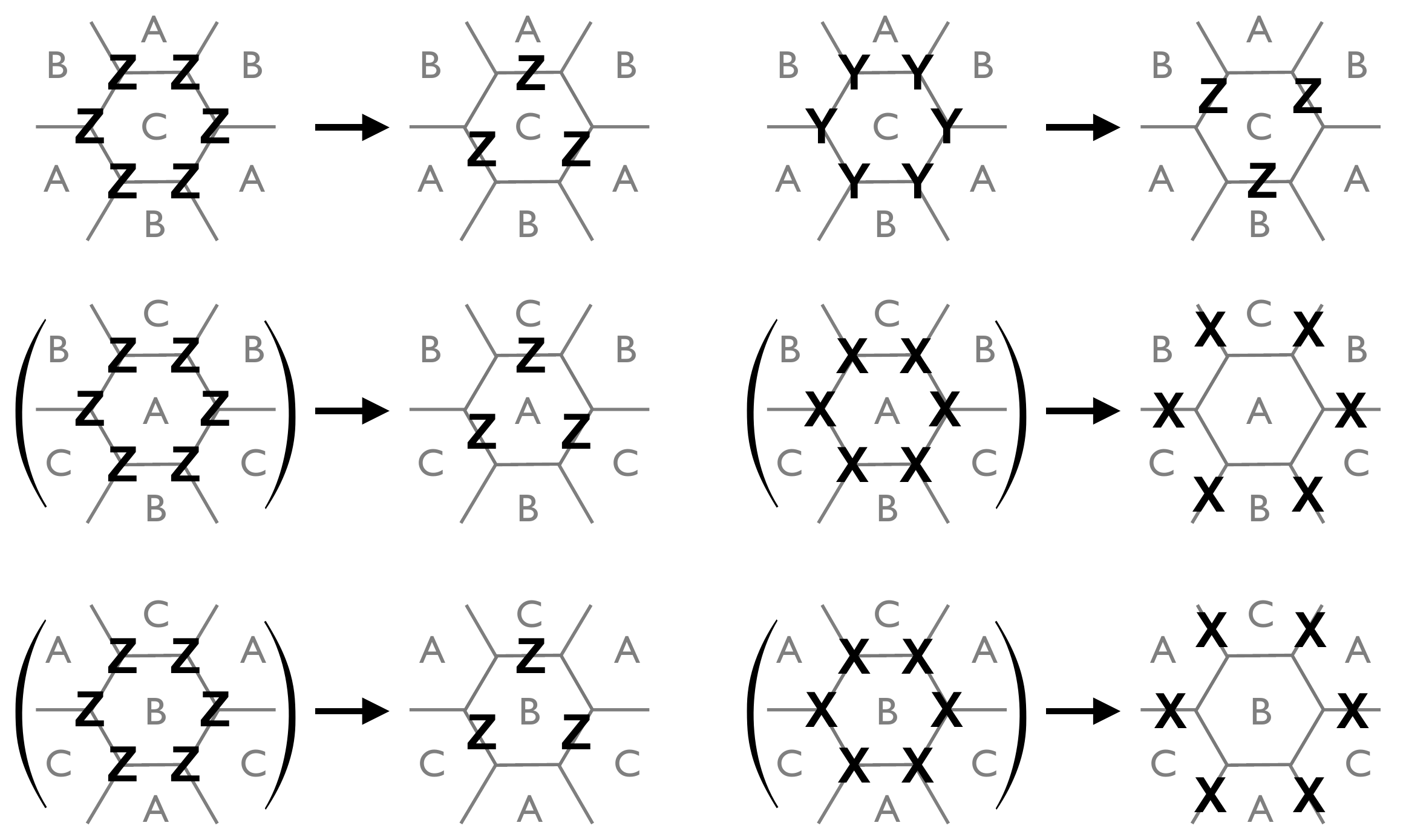}
\caption{The effect of applying the disentangling unitary transformation $U$ to the stabilizer group of the color code $CC(\mathcal{L})$. The parenthesis indicate that the stabilizer of the color code might be multiplied by the $X$-face stabilizers on neighboring faces of color $C$, depending on the disentangling procedure.}
\label{fig_2d_terms} 
\end{figure}

One can observe that after performing $U=\bigotimes_{f\in \mathcal{C}} U_f$, the $Z$-type stabilizers on faces of color $A$ and $C$, as well as the $X$-type stabilizers on faces of color $B$ transform into $Z$- and $X$-type stabilizers, respectively, on qubits on $AC$ edges. Similarly, the $Y$-type stabilizers on faces of color $C$, the $Z$-type stabilizers on faces of color $B$, and the $X$-type stabilizers on faces of color $A$ transform into stabilizers on qubits on $BC$ edges.

We conclude that after performing $U$, the stabilizer generators $CC(\mathcal{L})$ transform into two sets of stabilizer generators $TC(\mathcal{L}_A)$ and $TC(\mathcal{L}_B)$ supported on two disjoint sets of qubits, either placed on $BC$ or $AC$ edges.

\textbf{Step 3:}

We would like to show that the stabilizer generators $TC(\mathcal{L}_A)$ and $TC(\mathcal{L}_B)$ define the toric code on two lattices, $\mathcal{L}_A$ and $\mathcal{L}_B$, obtained from $\mathcal{L}$ by local deformations. A recipe for the shrunk lattice $\mathcal{L}_{A}$ is as follows:
\begin{itemize}
\item Vertices of $\mathcal{L}_{A}$ are centers of $A$ faces in $\mathcal{L}$. 
\item Edges of $\mathcal{L}_{A}$ are $BC$ edges in $\mathcal{L}$.
\item Faces of $\mathcal{L}_{A}$ are $B$ and $C$ faces in $\mathcal{L}$.
\end{itemize}
In short, one obtains $\mathcal{L}_{A}$ by shrinking $A$ faces to points while expanding $B$ and $C$ faces.~\cite{Bombin2007,Bombinbook}. Similarly, $\mathcal{L}_{B}$ is obtained by shrinking $B$ faces. Examples of shrunk lattices are depicted in Fig.~\ref{fig_2d_induced} for the case of the hexagonal lattice $\mathcal{L}$. In this case, one obtains two copies of the toric code supported on triangular lattices.

\begin{figure}[h!]
\includegraphics[width=0.85\textwidth]{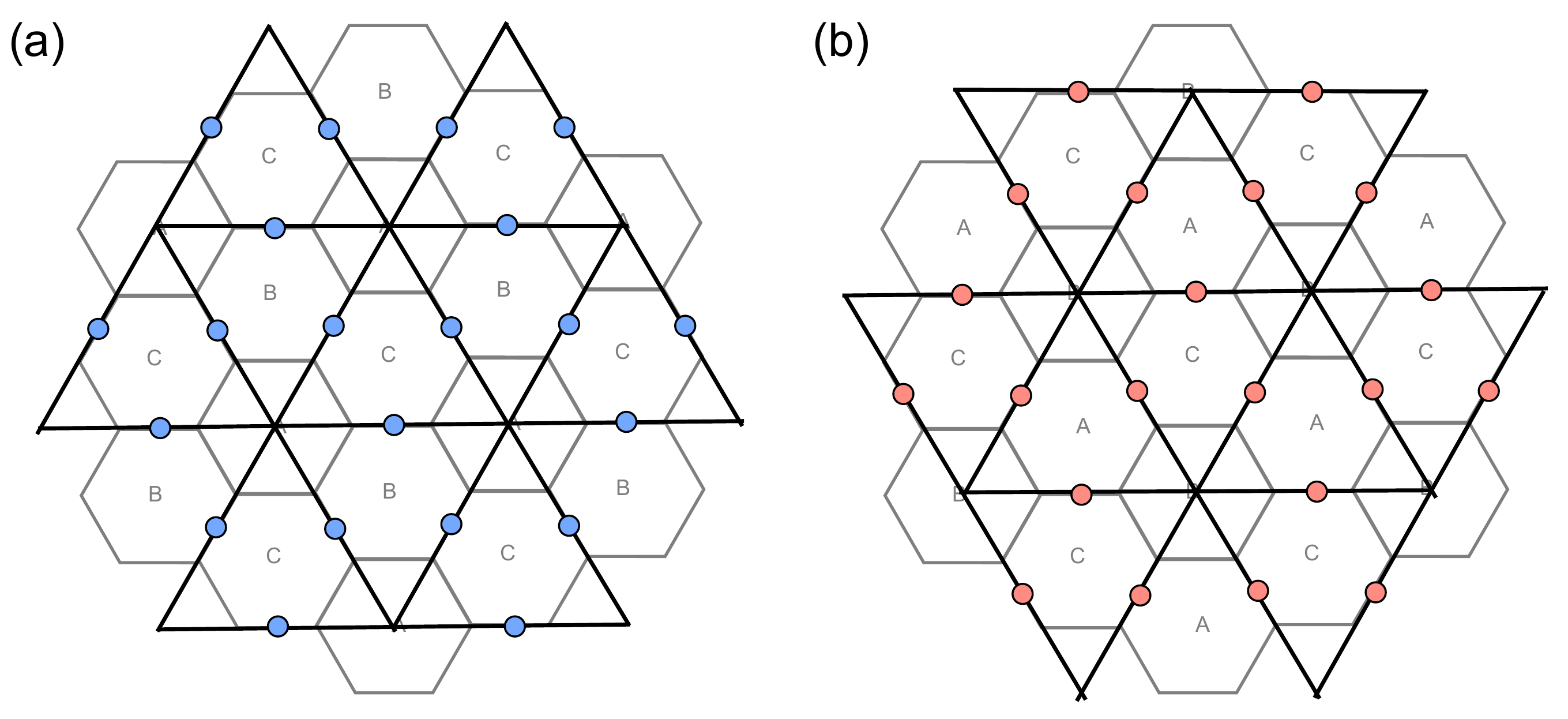}
\caption{(Color online) Fragments of the shrunk lattices: (a) $\mathcal{L}_{A}$ and (b) $\mathcal{L}_{B}$, obtained from $\mathcal{L}$ by shrinking $A$ and $B$ faces, respectively. Qubits are placed on edges, and the stabilizer generators are $X$-vertex and $Z$-face operators.}
\label{fig_2d_induced} 
\end{figure}

The stabilizer generators $TC(\mathcal{L}_A)$ and $TC(\mathcal{L}_B)$ are supported on either $\mathcal{L}_A$ or $\mathcal{L}_B$ lattices. In particular,
\begin{itemize}
\item the $X$-vertex stabilizers in $TC(\mathcal{L}_A)$ (respectively $TC(\mathcal{L}_B)$) are obtained from $X$-face stabilizers\footnote{\label{footnote2} Up to multiplication by $X$-face stabilizers on neighboring faces of color $C$.} of $CC(\mathcal{L})$ on $A$ (respectively $B$) faces,
\item the $Z$-face stabilizers in $TC(\mathcal{L}_A)$ (respectively $TC(\mathcal{L}_B)$) are obtained from $Z$-face stabilizers$^\mathrm{\ref{footnote2}}$ on $B$ faces (respectively $A$) and $Y$-face (respectively $Z$-face) stabilizers on $C$ faces.
\end{itemize}

To summarize, the unitary $U=\bigotimes_{f\in\mathcal{C}} U_f$ transforms the generators of the stabilizer group $CC(\mathcal{L})$ of the color code into two stabilizer groups $TC(\mathcal{L}_{A})$ and $TC(\mathcal{L}_{B})$, which define the toric code on two disjoint lattices $\mathcal{L}_{A}$ and $\mathcal{L}_{B}$ obtained from $\mathcal{L}$ by shrinking either $A$ or $B$ faces. This concludes the proof of the equivalence in two-dimensions.

Note that the equivalence between the two-dimensional color code and copies of the toric code has been proven for systems with translation symmetries~\cite{Beni10b,Bombin11}. Our work not only generalizes the previous results to the color code on an arbitrary lattice $\mathcal{L}$ on a closed manifold, but also presents an explicit construction of the local unitary and shrunk lattices. This leads to new observations for topological color codes with boundaries, which are presented in Section~\ref{sec:boundaries}.

\subsection{Isomorphism between Pauli subgroups}

In this subsection, we prove the existence of the disentangling unitary transformation $U=\bigotimes_{f\in\mathcal{C}} U_f$. We begin by developing some useful technical tool concerning properties of subgroups of the Pauli operator group. Consider a system of $n$ qubits and two subgroups of Pauli operators $\mathcal{O}_{1},\mathcal{O}_{2}\subseteq \mathbf{Pauli}(n)$, where $\mathbf{Pauli}(n)$ is the Pauli operator group on $n$ qubits. We shall neglect complex phases in $\mathcal{O}_{1},\mathcal{O}_{2}$. We say that $\mathcal{O}_{1}$ and $\mathcal{O}_{2}$ are \emph{isomorphic} to each other iff there exists a Clifford unitary transformation $U$ such that
\begin{align}
U\mathcal{O}_{1}U^{\dagger}=\mathcal{O}_{2}. 
\end{align}

Let $Z(\mathcal{O}_{1})$ and $Z(\mathcal{O}_{2})$ be centers of $\mathcal{O}_{1}$ and $\mathcal{O}_{2}$, respectively. Then, the following lemma holds~\cite{Beni10}:
\begin{lemma}[Isomorphic Groups]
\label{lemma:isomorphicgroups}
Two subgroups of Pauli operators $\mathcal{O}_{1},\mathcal{O}_{2}\subset \mathbf{Pauli}(n)$ are isomorphic iff 
\begin{align}
G(\mathcal{O}_{1})=G(\mathcal{O}_{2}),\qquad G(Z(\mathcal{O}_{1}))=G(Z(\mathcal{O}_{2})),
\end{align}
where $G(\mathcal{O})$ represents the number of independent generators of $\mathcal{O}\subset \mathbf{Pauli}(n)$. 
\end{lemma}

Let $\{g_{j}\}$ and $\{h_{j}\}$ be two sets of independent generators for two isomorphic groups $\mathcal{O}_{1}$ and $\mathcal{O}_{2}$. We say that $\{g_{j}\}$ and $\{h_{j}\}$ have the same commutation relations if
\begin{align}
g_{i}g_{j}=(-1)^{c_{i,j}}g_{j}g_{i}, \qquad h_{i}h_{j}=(-1)^{c_{i,j}}h_{j}h_{i} \qquad (c_{i,j}=0,1). 
\end{align}
Note  that $c_{i,j}=c_{j,i}$ and $c_{i,i}=0$. We have the following lemma.

\begin{lemma}[Clifford Transformation]
\label{lemma:Cliffordunitary}
Let $\mathcal{O}_{1}$ and $\mathcal{O}_{2}$ be two isomorphic groups generated by two sets of independent generators, $\{g_{j}\}$ and $\{h_{j}\}$. If $\{g_{j}\}$ and $\{h_{j}\}$ have the same commutation relations, then there exists a Clifford unitary transformation $U$ such that
\begin{equation}
Ug_{j}U^{\dagger}=h_{j} \qquad \forall j.
\end{equation}
\end{lemma}

\begin{proof}
Let us find a canonical set of independent generators for $\mathcal{O}_{1}$:
\begin{align}
\mathcal{O}_{1}=\left\langle
\begin{array}{ccccccccc}
A_{1},  &  \ldots,& A_{n_{1}},  & A_{n_{1}+1} ,& \ldots,& A_{n_{2}}  \\
A_{n_{2}+1},  & \ldots,& A_{n_{1}+n_{2}}
\end{array}
\right\rangle ,
\end{align}
where $n_{2}\geq n_{1}$, and two Pauli operators $A_i$ and $A_j$ commute unless they are in the same column, in which case they anti-commute by definition. Note that any canonical generator can be written as a product of generators $\{ g_{j} \}$.

For a binary vector $\vec{a}=(a_1,\ldots,a_{n_{1}+n_{2}})$, we define
\begin{align}
\mathcal{O}_{1}(\vec{a})=\prod_{j=1}^{n_{1}+n_{2}} g_{j}^{a_{j}},\qquad
\mathcal{O}_{2}(\vec{a})=\prod_{j=1}^{n_{1}+n_{2}} h_{j}^{a_{j}}.
\end{align}
Then, there exists a set of independent $n_{1}+n_{2}$ binary vectors $\vec{a}^{(i)}$ such that
\begin{align}
A_{i}=\mathcal{O}_{1}(\vec{a}^{(i)}). 
\end{align}
Let $B_{j}=\mathcal{O}_{2}(\vec{a}^{(j)})$. Since commutation relations of $\{g_{j}\}$ and $\{h_{j}\}$ are identical, then $B_{j}$ are canonical generators for $\mathcal{O}_{2}$:
\begin{align}
\mathcal{O}_{2}=\left\langle
\begin{array}{ccccccccc}
B_{1},  &  \ldots,& B_{n_{1}},  & B_{n_{1}+1}, &\ldots,& B_{n_{2}}  \\
B_{n_{2}+1},  & \ldots,& B_{n_{1}+n_{2}}
\end{array}
\right\rangle
\end{align}
Then, as shown in Ref.~\cite{Beni10}, there exists a Clifford unitary $U$ such that 
\begin{align}
UA_{j}U^{\dagger}=B_{j}\qquad \forall j\in\{1,\ldots,n_1+n_2 \}.
\end{align}
Such a unitary transformation also satisfies
\begin{align}
Ug_{j}U^{\dagger}=h_{j}\qquad \forall j\in\{ 1,\ldots,n_1+n_2\},
\end{align}
which completes the proof of the (Clifford Transformation) Lemma~\ref{lemma:Cliffordunitary}
\end{proof}

\begin{figure}[h!]
\includegraphics[width=0.35\textwidth]{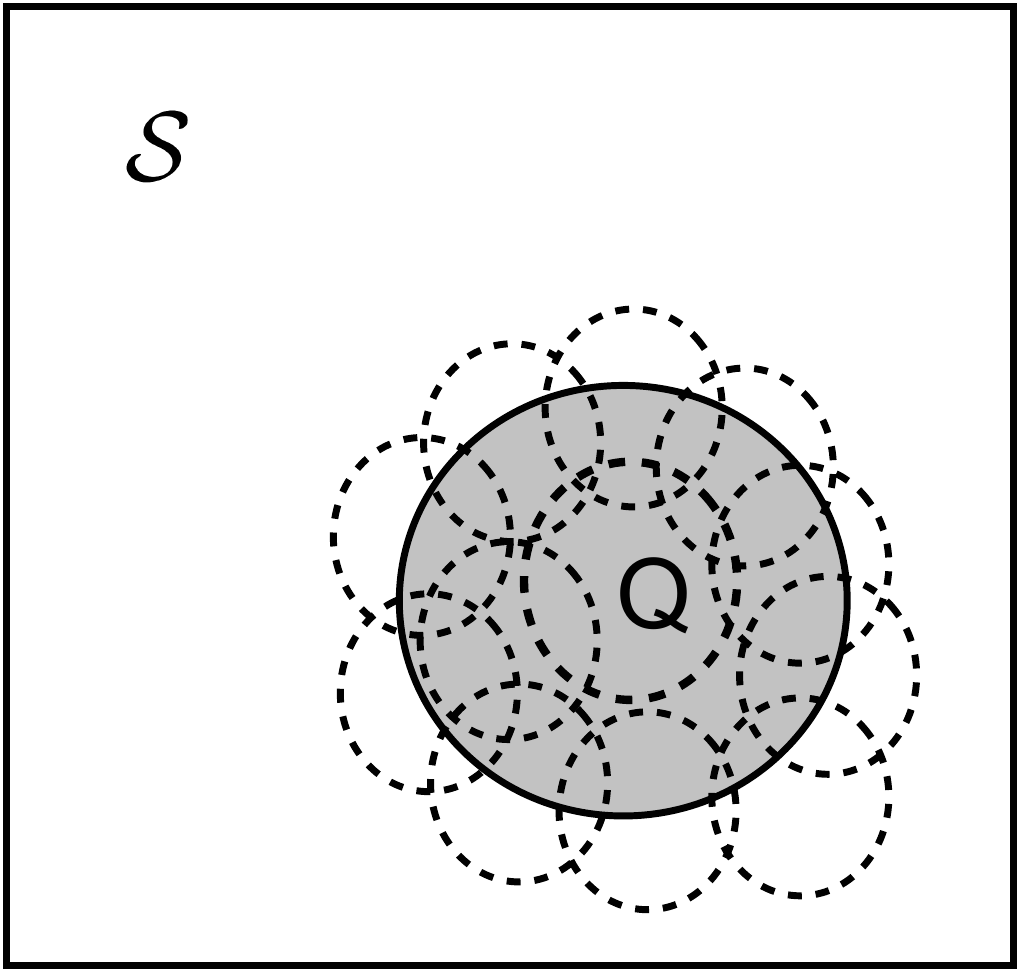}
\caption{The overlap group of the stabilizer group $\mathcal{S}$ on the region $Q$ is defined as the group generated by the restriction of the generators of $\mathcal{S}$ onto $Q$. Dotted circles represent the stabilizer generators of $\mathcal{S}$ with support intersecting $Q$.}
\label{fig_overlap} 
\end{figure} 

We are ready to show the existence of a Clifford unitary $U_f$, which satisfies the rules in Eqs.~(\ref{eq:list1})--(\ref{eq:list2}). First, let us introduce the notion of the overlap group of the stabilizer group \cite{Beni10}. For a given subset of qubits, denoted by $Q$, the overlap group on $Q$ is defined as the group generated by the restriction of generators of the stabilizer group $\mathcal{S}$ onto $Q$. Namely, 
\begin{align}
\mathcal{O}_{Q}= \big\langle u|_{Q}\ \big| u \in \mathcal{S}  \big\rangle,
\end{align}
where $u|_{Q}$ represents a restriction of $u$ onto $Q$ (see Fig.~\ref{fig_overlap}). Note that the overlap group is not necessarily Abelian and is defined up to a global phase.

The key idea in the proof of existence of $U$ is that the overlap groups for the color code and the toric code for the set of $C$ faces are isomorphic. In particular, let us consider a $C$ face $f\in \mathcal{L}$ with $2n$ vertices, and two corresponding faces $f^A \in\mathcal{L}_A$ and $f^B \in\mathcal{L}_B$ derived from $f$. Then, the overlap group of $CC(\mathcal{L})$ on $f$ is generated by
\begin{align}
\mathcal{O}_f=\langle Z_{j}Z_{j+1}, X_{j}X_{j+1}\ | j\in\{1,\ldots,2n\}\rangle,
\end{align}
whereas the overlap group of $TC(\mathcal{L}_A)$ and $TC(\mathcal{L}_B)$ on $f^A \sqcup f^B$ is generated by 
\begin{align}
\mathcal{O}_{f^A \sqcup f^B}=\langle Z_{(j,j+1)}, X_{(j-1,j)}X_{(j+1,j+2)}\ | j\in\{1,\ldots,2n\} \rangle. 
\end{align}
Observe that both $\mathcal{O}_f$ and $\mathcal{O}_{f^A \sqcup f^B}$ have $4n-2$ independent generators and their centers are generated by $2$ independent operators. Namely,
\begin{eqnarray}
G(\mathcal{O}_f) = G(\mathcal{O}_{f^A \sqcup f^B}),\\
G(Z(\mathcal{O}_f))=G(Z(\mathcal{O}_{f^A \sqcup f^B})).
\end{eqnarray}
Using the (Isomorphic Groups) Lemma~\ref{lemma:isomorphicgroups}, we obtain that $\mathcal{O}_f$ and $\mathcal{O}_{f^A \sqcup f^B}$ are isomorphic.

Let us choose a set of independent generators for $\mathcal{O}_{f}$ as follows
\begin{eqnarray}
\label{eq:list3}
g_{j} =& Z_{j}Z_{j+1} \qquad\qquad\qquad &(j=1,\ldots,2n-1),\\
g_{2n} =& \left(\bigotimes_{i=1}^{2n} X_i\right) Z_{2n}Z_{1}\qquad&\\
g_{j+2n} =& X_{j}X_{j+1} \qquad\qquad\qquad &(j=1,\ldots,2n-2).
\end{eqnarray}
We then label a set of independent generators for $\mathcal{O}_{f^A \sqcup f^B}$ in the following way
\begin{eqnarray}
h_{j} =&\ Z_{(j,j+1)}\qquad\qquad\qquad &(j=1,\ldots,2n),\\
h_{j+2n} =&\ X_{(j-1,j)}X_{(j+1,j+2)} \qquad &(j=1,\ldots,2n-2).
\label{eq:list4}
\end{eqnarray}
By direct calculation one can verify that $\{g_{j}\}$ and $\{h_{j}\}$ have the same commutation relations. Thus, from the (Clifford Transformation) Lemma~\ref{lemma:Cliffordunitary}, there exists a Clifford unitary $U_f$ such that
\begin{align}
U_f g_{j}U_f^{\dagger}=h_{j}\qquad \forall j\in\{1,\ldots,4n-2 \}.
\end{align}
Therefore, the local Clifford unitary $U=\bigotimes_{f\in\mathcal{C}} U_f$ transforms $CC(\mathcal{L})$ into $TC(\mathcal{L}_A)\otimes TC(\mathcal{L}_B)$, and this concludes the proof of the Theorem~\ref{th:equivin2d}.

One might find the labelings in Eqs.~(\ref{eq:list3})--(\ref{eq:list4}) arbitrary. Yet, once we have chosen $g_{j}$ for $j=1,\ldots,2n$, it is not difficult to find the right labeling for $j=2n+1,\ldots,4n-2$ by checking the commutation relations. Note that the choice of $g_{2n} = \left(\bigotimes_{j=1}^{2n} X_j \right) Z_{2n}Z_1$ is crucial to ensure that the generators $\{ g_{j} \}_{j=1}^{2n}$ are independent.

\subsection{Three (or more) dimensions}

A similar equivalence between the topological color code and the toric code holds in any dimensions. It can be summarized in the following theorem.
\begin{theorem}[Equivalence]
\label{th:equivalence}
Let $CC(\mathcal{L})$ be the stabilizer group of the topological color code defined on a $d$-dimensional lattice $\mathcal{L}$ without boundaries, which is $(d+1)$-valent and colored with $C_0,\ldots,C_d$. Let $X$- and $Z$-type stabilizer generators be supported on $d$-cells and $2$-cells, where $d\geq 2$. Then, there exists a local Clifford unitary $U$ such that
\begin{equation}
U[CC(\mathcal{L}) \bigotimes \mathcal{S}] U^{\dagger}= \bigotimes_{j=1}^{d} TC(\mathcal{L}_j),
\end{equation}
where $\mathcal{S}$ represents the stabilizer group of decoupled ancilla qubits, and $TC(\mathcal{L}_j)$ -- the stabilizer group of the toric code defined on the shrunk lattice $\mathcal{L}_j$ derived from $\mathcal{L}$ by local deformations, i.e. shrinking $d$-cells of color $C_j$. Moreover, one can choose the disentangling unitary $U$ to be of the form
\begin{equation}
U= \bigotimes_{c \in \mathcal{C}_0} U_{c},
\end{equation}
where $\mathcal{C}_0$ is the set of $d$-cells of color $C_0$ in $\mathcal{L}$, and $U_{c}$ is a Clifford unitary acting only on qubits on vertices of the $d$-cell $c$. 
\end{theorem}

Note that the color code qubits are placed on vertices, whereas the toric code qubits --- on edges. Thus, for every $d$-cell $c$ colored in $C_0$, we shall add $E-V$ ancilla qubits, where $V$ and $E$ denote the number of vertices and edges in $c$. We can assume that ancilla qubits are stabilized by single-qubit Pauli $Z$ operators. Since the lattice $\mathcal{L}$ is $(d+1)$-valent, then $E=dV/2$ and $E-V\geq 0$ for $d\geq 2$. In particular, ancilla qubits are required for the three- or higher-dimensional case.

Since the color code and the toric code in $d$ dimensions support anyonic excitations whose braiding properties are similar (there exists an isomorphism between anyon labels for two codes), the equivalence should not be very surprising. Yet, our result may be of interest from the viewpoint of finding topological invariants to classify topological phases. It has been argued that two topologically ordered systems with isomorphic anyon labels and modular matrices belong to the same topological phase~\cite{Chen10,Hastings05,Levin05}. 
This hypothesis has been proven for two-dimensional stabilizer Hamiltonians with translation symmetries~\cite{Beni10b}. Also, this hypothesis has been tested for the two-dimensional Levin-Wen model in Ref.~\cite{Kitaev12}, where a construction of a transparent domain wall between two Levin-Wen models (with tensor unitary categories satisfying certain equivalence conditions) was presented.

The idea of the mapping is a straightforward generalization of the proof of the Theorem~\ref{th:equivin2d} presented in Section~\ref{sec:closed}B. First, we perform a local Clifford unitary, whose existence is guaranteed by the (Clifford Transformation) Lemma~\ref{lemma:Cliffordunitary}. Then, we analyze how the stabilizer generators of the color code transform under such a unitary. Finally, we check that the stabilizers can be split into $d$ sets, each of them defining a copy of the toric code on a lattice obtained by deforming the initial lattice $\mathcal{L}$. For the sake of clarity, we focus on $d=3$. We also first present the construction of shrunk lattices, before explaining how to construct a local Clifford unitary transforming the color code into $d$ decoupled copies of the toric code.

In three dimensions, the lattice $\mathcal{L}$ has volumes colored with four colors, $A$, $B$, $C$ and $D$. Recall that we can assign colors to faces and edges, too. Namely, a face has two colors of two volumes it belongs to, whereas an edge has three colors (of three volumes it belongs to). We obtain three shrunk lattices, $\mathcal{L}_A$, $\mathcal{L}_B$ and $\mathcal{L}_C$  by shrinking volumes of color $A$, $B$ and $C$, respectively. In particular, $\mathcal{L}_A$ consists of
\begin{itemize}
\item vertices --- centers of $A$ volumes in $\mathcal{L}$, 
\item edges --- $BCD$ edges in $\mathcal{L}$,
\item faces --- $BC$, $BD$ and $CD$ faces in $\mathcal{L}$,
\item volumes --- $B$, $C$ and $D$ volumes in $\mathcal{L}$.
\end{itemize}
For an example, see Fig.~\ref{fig_3d_surface}. Similarly for other shrunk lattices $\mathcal{L}_B$ and $\mathcal{L}_C$. In general, a $d$-dimensional lattice $\mathcal{L}$ is colored with $d+1$ colors, $C_0,\ldots,C_d$, and one obtains the shrunk lattice $\mathcal{L}_i$, where $i=1,\ldots,d$, by shrinking $d$-cells of color $C_i$. Namely, $\mathcal{L}_i$ consists of
\begin{itemize}
\item vertices --- centers of $d$-cells in $\mathcal{L}$ of color $C_i$,
\item edges --- edges in $\mathcal{L}$ of color $\{C_0,\ldots,C_d\}\setminus \{ C_i\}$,
\item faces --- faces in $\mathcal{L}$ of color $\{C_0,\ldots,C_d\}\setminus \{ C_i,C_j\}$ for all $j\neq i$.
\end{itemize}

\begin{figure}[h!]
\includegraphics[width=0.75\textwidth]{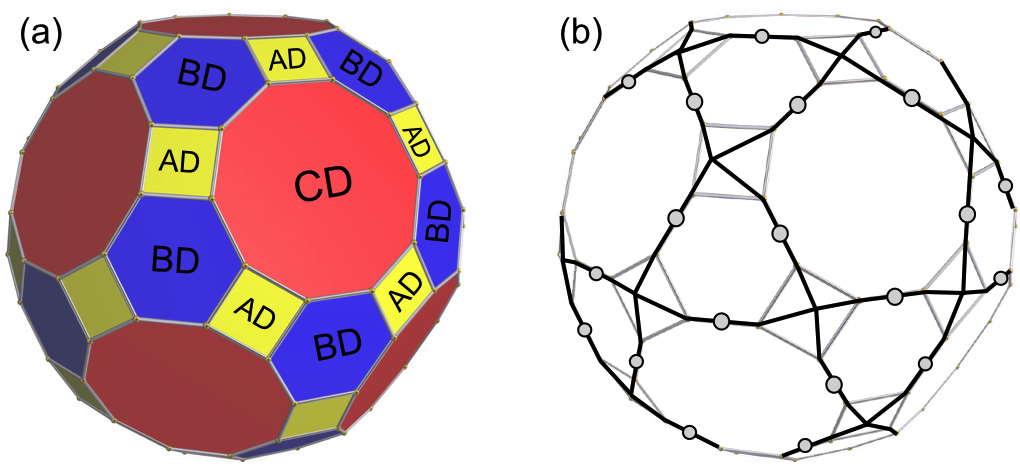}
\caption{(Color online) (a) The boundary $\partial c$ of a volume $c$ of color $D$ in the lattice $\mathcal{L}$. Note that $\partial c$ can be viewed as a $3$-colorable and $3$-valent lattice on a closed manifold (a sphere), with faces colored in $AD$, $BD$ and $CD$. (b) A volume in the shrunk lattice $\mathcal{L}_A$ derived from $c$ after shrinking volumes of color $A$. Note that qubits are placed on (a) vertices and (b) edges. The figures were created using Robert Webb's Stella software (http://www.software3d.com/Stella.php).}
\label{fig_3d_surface} 
\end{figure}

We construct the disentangling unitary $U$ as a tensor product of local Clifford unitaries, $U=\bigotimes_{c\in\mathcal{D}} U_c$, where $\mathcal{D}$ is the set of all volumes of color $D$. Let us consider a volume $c$ of color $D$. The overlap group $\mathcal{O}_{CC}$ of the stabilizer group of the color code on $c$ is generated by $Z$-edge operators and $X$-face operators, for each and every edge and face belonging to $c$. Namely, 
\begin{equation}
\mathcal{O}_{CC} = \left\langle \raisebox{-15pt}{\includegraphics[height=0.5in]{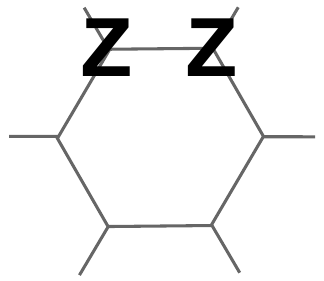}} , \raisebox{-15pt}{\includegraphics[height=0.5in]{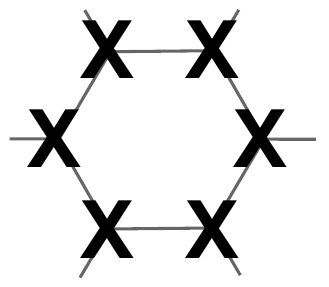}} \right\rangle.
\end{equation}
Let $\mathcal{H}_V\simeq (\mathbb{C}^2)^{\otimes V}$ and $\mathcal{H}_E\simeq (\mathbb{C}^2)^{\otimes E}$ be the Hilbert spaces of qubits placed on vertices and edges, respectively. Since $E-V>0$, we need to add $E-V$ ancilla qubits to qubits on vertices to match the dimensionality of Hilbert spaces, $\mathcal{H}_V \otimes\mathcal{H}_{ancilla}\simeq \mathcal{H}_E$, where $\mathcal{H}_{ancilla}$ is the Hilbert space of ancilla qubits. Let $\mathcal{S}_c = \langle Z_i|\ \forall i\in\{1,\ldots,E-V\}\rangle$ be the stabilizer group of the ancilla qubits, where $Z_i$ is the Pauli $Z$ operator acting on the ancilla qubit $i$. We would like to construct a Clifford unitary $U_c$ which maps the group $\mathcal{O}_{CC}\otimes\mathcal{S}_c$ of operators on the Hilbert space $\mathcal{H}_V\otimes \mathcal{H}_{ancilla}$ into the group 
\begin{equation}
\mathcal{O}_{TC} = \left\langle \raisebox{-15pt}{\includegraphics[height=0.5in]{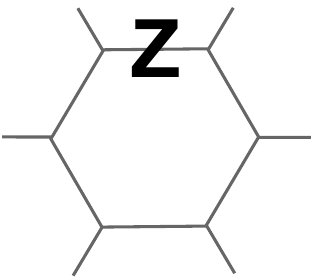}} , \raisebox{-15pt}{\includegraphics[height=0.5in]{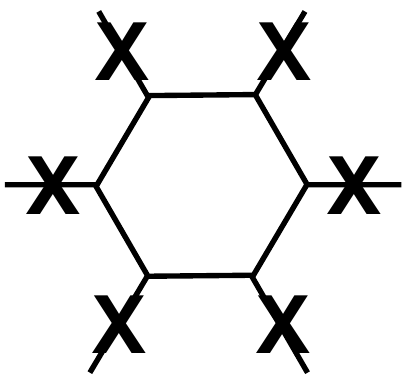}}\right\rangle
\end{equation}
of operators on $\mathcal{H}_E$ according to the rules
\begin{equation}
\label{eq:ruleforUf}
\left( \raisebox{-15pt}{\includegraphics[height=0.5in]{eq_Zedge.pdf}}  \right) \rightarrow 
\raisebox{-15pt}{ \includegraphics[height=0.5in]{eq_Zedge_t.pdf}}, \qquad 
\left( \raisebox{-15pt}{\includegraphics[height=0.5in]{eq_plaquette.pdf}} \right) \rightarrow 
\raisebox{-15pt}{ \includegraphics[height=0.5in]{eq_plaquette_t.pdf}} .
\end{equation}
The parenthesis indicate that the mapping holds up to multiplication by the elements of the center $Z(\mathcal{O}_{CC}\otimes \mathcal{S}_c)$.

Let us analyze what happens to the stabilizer group of the color code and and the stabilizer group of ancilla qubits, $CC(\mathcal{L})\bigotimes_{c\in\mathcal{D}} \mathcal{S}_c$, after applying the unitary $U=\bigotimes_{c\in\mathcal{D}} U_c$. One can verify that
\begin{itemize}
\item $X$-vertex stabilizers of $TC(\mathcal{L}_A)$, $TC(\mathcal{L}_B)$ and $TC(\mathcal{L}_C)$ are obtained from $X$-volume stabilizers\footnote{\label{footnote} Up to multiplication by elements of the center $Z(\mathcal{O}_{CC}\otimes\mathcal{S}_c)$ for any neighboring volume $c$ of color $D$.} in $CC(\mathcal{L})$ of color $A$, $B$ and $C$, respectively,
\item $Z$-face stabilizers in $TC(\mathcal{L}_A)$ are obtained from $Z$-face stabilizers$^\mathrm{\ref{footnote}}$ of color $BD$, $CD$ and $BC$; similarly for $TC(\mathcal{L}_B)$ and $TC(\mathcal{L}_C)$,
\item the elements in the center $Z(\mathcal{O}_{CC}\otimes\mathcal{S}_c)$ are mapped into the center $Z(\mathcal{O}_{TC})$.
\end{itemize}
Moreover, the generators of the group $U\left(CC(\mathcal{L})\bigotimes_{c\in\mathcal{D}} \mathcal{S}_c\right ) U^\dag$ are supported on either $\mathcal{L}_A$, or $\mathcal{L}_B$, or $\mathcal{L}_C$, and thus one obtains three decoupled copies of the toric code.

The last thing we need to justify is the existence of $U_c$ consistent with the rules in Eq.~(\ref{eq:ruleforUf}). We start with showing that $\mathcal{O}_{CC}$ and $\mathcal{O}_{TC}$ are isomorphic. Clearly, $\mathcal{O}_{CC}, \mathcal{O}_{TC}\subset \mathbf{Pauli}(n=E)$. First, let us look at the independent generators of $\mathcal{O}_{CC}$. Note that there are $V-1$ independent operators of type $\raisebox{-15pt}{\includegraphics[height=0.5in]{eq_Zedge.pdf}}$, denoted by  $\{g_i\}_{i=1}^{V-1}$, supported on edges of a spanning tree $T\subset E$ of the graph $G=(V,E)$. In the case of operators of type $\raisebox{-15pt}{\includegraphics[height=0.5in]{eq_plaquette.pdf}}$, there are exactly two independent relations between them, namely
\begin{equation}
\prod_{f\in \mathcal{AD}} \raisebox{-15pt}{\includegraphics[height=0.5in]{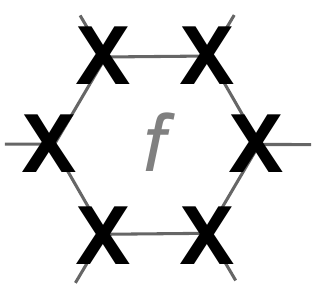} }= \prod_{f\in \mathcal{BD}} \raisebox{-15pt}{\includegraphics[height=0.5in]{eq_plaquettef.pdf}} = \prod_{f\in \mathcal{CD}} \raisebox{-15pt}{\includegraphics[height=0.5in]{eq_plaquettef.pdf}},
\end{equation}
where the products are taken over all $X$-face operators associated with faces of $c$ of color $AD$, $BD$, and $CD$, respectively. Thus, there are $F-2$ independent $X$-face operators. We set $F-3$ generators $\{g_i\}_{i=V}^{V+F-4}$ to be $X$-face operators, associated with all faces of $c$ but three --- one of each color $AD$, $BD$ and $CD$. We also set $g_{V+F-3}= \bigotimes_{v\in V} X(v)$, where $\bigotimes_{v\in V} X_v$ is the $X$-volume operator on $c$. Including $E-V$ single qubit Pauli $Z$ stabilizer generators $\{g_i\}_{i=V+F-2}^{E+F-3}$ for ancilla qubits, there are
\begin{equation}
(V-1)+(F-2)+(E-V)= E+F-3
\end{equation}
independent generators of $\mathcal{O}_{CC}$, and thus $G(\mathcal{O}_{CC})=E+F-3$. Note that since
\begin{equation}
Z(\mathcal{O}_{CC}) = \left\langle \raisebox{-15pt}{\includegraphics[height=0.5in]{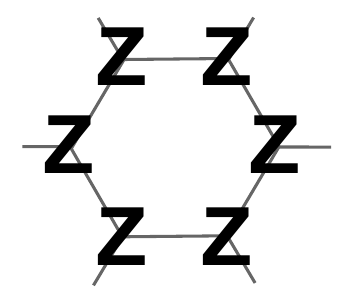}}, \bigotimes_{v\in V} X(v), Z_i \right\rangle,
\end{equation}
then $G(Z(\mathcal{O}_{CC}))= (F-2) +1 + (E-V)$.

In the case of $\mathcal{O}_{TC}$, there are $E$ independent generators of type $\raisebox{-15pt}{\includegraphics[height=0.5in]{eq_Zedge_t.pdf}}$. Observe that there are only three independent relations between generators of type $\raisebox{-15pt}{\includegraphics[height=0.5in]{eq_plaquette_t.pdf}}$, namely
\begin{equation}
\prod_{f\in \mathcal{AD}} \raisebox{-15pt}{\includegraphics[height=0.5in]{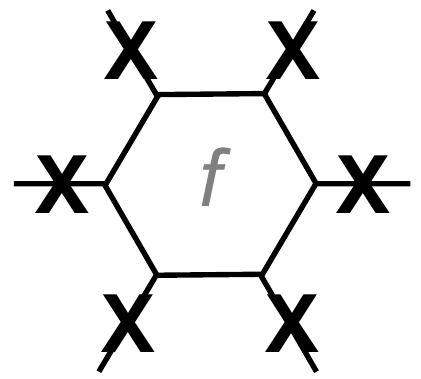} }= \prod_{f\in \mathcal{BD}} \raisebox{-15pt}{\includegraphics[height=0.5in]{eq_plaquette_tf.pdf}} = \prod_{f\in \mathcal{CD}} \raisebox{-15pt}{\includegraphics[height=0.5in]{eq_plaquette_tf.pdf}}=I,
\end{equation}
and thus $G(\mathcal{O}_{TC})=E+F-3$. Since the group $\mathcal{O}_{TC}$ has single qubit Pauli Z operators as generators, the center $Z(\mathcal{O}_{TC})$ can only be generated by $Z$-type operators,
\begin{equation}
Z(\mathcal{O}_{TC}) = \left\langle \raisebox{-15pt}{\includegraphics[height=0.5in]{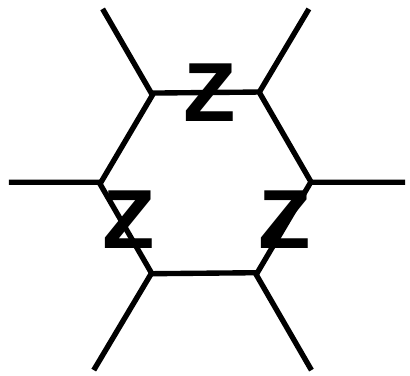}} \right\rangle.
\end{equation}
There are $2F$ operators of type $\raisebox{-15pt}{\includegraphics[height=0.5in]{fig_Zfacegen1.pdf}}$, and they satisfy three independent relations, namely a product of all $Z$-face operators with qubits placed on edges not colored in $i$, for $i\in\{A,B,C\}$. Thus $G(Z(\mathcal{O}_{TC}))=2F-3$ and using Euler characteristic for $c$, $V-E+F=2$, we obtain $G(Z(\mathcal{O}_{TC}))=G(Z(\mathcal{O}_{CC}))$. From the (Isomorphic Groups) Lemma~\ref{lemma:isomorphicgroups} we obtain that $\mathcal{O}_{CC}$ and $\mathcal{O}_{TC}$ are isomorphic.

We have already chosen independent generators $\{ g_i\}$ of $\mathcal{O}_{CC}$. We choose independent generators $\{ h_i\}$ of $\mathcal{O}_{TC}$ as follows:
\begin{itemize}
\item for $i=1,\ldots, V-1$: $g_i = \raisebox{-15pt}{\includegraphics[height=0.5in]{eq_Zedge.pdf}} \rightarrow  h_i = 
\raisebox{-15pt}{ \includegraphics[height=0.5in]{eq_Zedge_t.pdf}}$,
\item for $i=V,\ldots, V+F-4$: $g_i = \raisebox{-15pt}{\includegraphics[height=0.5in]{eq_plaquette.pdf}} \rightarrow  h_i =
\raisebox{-15pt}{ \includegraphics[height=0.5in]{eq_plaquette_t.pdf}}$,
\item for $i=V+F-3$: $g_i = \bigotimes_{v\in V} X(v) \rightarrow  h_i =
\bigotimes_{e\in E} Z(e)$,
\item for $i=V+F-2,\ldots,E+F-3$: $g_i = Z_i \rightarrow h_i\in Z(\mathcal{O}_{TC})$,
\end{itemize}
where $Z_i$ is a Pauli $Z$ operator on the ancilla qubit. We would like to emphasize that the choice of $\{h_i\}^{E+F-3}_{i=V+F-2}$ does not matter, as long as they belong to the center $Z(\mathcal{O}_{TC})$ and $\{h_i\}_{i=1}^{E+F-3}$ is the set of independent operators. One can verify that $\{g_i\}$ and $\{h_i\}$ have the same commutation relations, and thus from the (Clifford Transformation) Lemma~\ref{lemma:Cliffordunitary}, there exists a Clifford unitary $U_c$ such that
\begin{equation}
U_c g_i U_c^\dag = h_i\qquad \forall i\in\{ 1,\ldots,E+F-3\}. 
\end{equation}
Moreover, the choice of generators $\{g_i\}$ and $\{h_i\}$ guarantees that the rules in Eq.~(\ref{eq:ruleforUf}) are satisfied. This concludes the proof of the (Equivalence) Theorem~\ref{th:equivalence} in $d=3$ dimensions.

Finally, we present a sketch of a proof for higher-dimensional case (for a rigorous proof, see the Appendix). The color code is defined on a $d$-dimensional lattice $\mathcal{L}$ with $d$-cells colored in $C_0,C_1,\ldots,C_d$. Let $c$ be a $d$-cell in $\mathcal{L}$ of color $C_0$ with $V$ vertices, $E$ edges and $F$ $(d-1)$-cells. Let $\mathcal{O}_{CC}$ be the overlap group of the stabilizer group $CC(\mathcal{L})$ of the color code on $c$ and $\mathcal{S}_c$ be the stabilizer group of $E-V$ ancilla qubits. Note that $\mathcal{O}_{CC}\otimes\mathcal{S}_c$ is generated by $Z$-edge operators, $X$-type $(d-1)$-cell-like operators and single Pauli $Z$ operators on ancilla qubits. Thus, $G(\mathcal{O}_{CC}\otimes\mathcal{S}_c) = (V-1) + (F-d) + (E-V) = E+F-d$. Let $\mathcal{O}_{TC}$ be defined as a group of operators on qubits placed on edges of $c$. Namely, $\mathcal{O}_{TC}$ is generated by single qubit Pauli $Z$ operators on edges and $X$-vertex-like operators with support on all edges radiating out of $(d-1)$-cells of $c$. Note that there are $d$ independent relations between $X$-vertex-like operators, namely a product of all $X$-vertex-like operators associated with $(d-1)$-cells of certain color is identity. Thus, $G(\mathcal{O}_{TC}) = E+F-d$. By relating the number of independent generators of $Z(\mathcal{O}_{CC}\otimes\mathcal{S}_c)$ and $Z(\mathcal{O}_{TC})$ to the number of $i$-cells of $c$, for $i=0,1,\ldots,d$, and the Betti numbers of $c$, we can prove $G(Z(\mathcal{O}_{CC}\otimes\mathcal{S}_c))=G(Z(\mathcal{O}_{TC}))$ (see the Appendix for more details). From the (Isomorphic Groups) Lemma~\ref{lemma:isomorphicgroups} we obtain that $\mathcal{O}_{CC}\otimes \mathcal{S}_c$ and $\mathcal{O}_{TC}$ are isomorphic. We then choose independent generators $\{g_i\}$ and $\{h_i\}$ of $\mathcal{O}_{CC}\otimes \mathcal{S}_c$ and $\mathcal{O}_{TC}$ as follows
\begin{itemize}
\item $\{g_i\}_{i=1}^{V-1}$ --- independent $Z$-edge operators related to a spanning tree $T\subset E$ of the graph $G=(V,E)$ 
$\rightarrow \{h_i\}_{i=1}^{V-1}$ --- single qubit Pauli $Z$ operators on qubits placed on edges associated with the spanning tree $T$,
\item $\{g_i\}_{i=V}^{F+V-d-1}$ --- independent $X$-type $(d-1)$-cell operators associated with all $(d-1)$-cells of $c$ except for $d$ of them, namely one $(d-1)$-cell for each colors $C_0 C_1$, $C_0 C_2,\ldots,C_0 C_d\rightarrow \{h_i\}_{i=V}^{F+V-d-1}$ --- $X$-vertex-like operators with support on edges radiating out of $F-d$ corresponding $(d-1)$-cells of $c$,
\item $g_{i=F+V-d}=\bigotimes_{v\in V} X(v)\rightarrow h_{i=F+V-d}\bigotimes_{e\in E} Z(e)$,
\item $\{g_i\}_{i=F+V-d+1}^{E+F-d}$ --- single Pauli $Z$ operators on ancilla qubits $\rightarrow$ $\{h_i\}_{i=F+V-d+1}^{E+F-d}\in Z(\mathcal{O}_{TC})$ --- elements of the center of $\mathcal{O}_{TC}$ chosen in such a way that all the operators $\{h_i\}$ are independent.
\end{itemize}
One can verify that $\{g_i\}$ and $\{h_i\}$ have the same commutation relations. From the (Clifford Transformation) Lemma~\ref{lemma:Cliffordunitary}, there exists of a local Clifford unitary $U_c$ such that
\begin{equation}
U_c g_i U_c^\dag = h_i \qquad\forall i\in\{1,\ldots, E+F-d\}.
\end{equation}
By applying the disentangling unitary transformation $U=\bigotimes_{c\in \mathcal{C}_0} U_c$ to the stabilizer group $CC(\mathcal{L})$ of the color code and the stabilizer group $\mathcal{S}=\bigotimes_{c\in\mathcal{C}_0} \mathcal{S}_c$ of ancilla qubits, one obtains the stabilizer groups of the toric code supported on $d$ decoupled lattices $\mathcal{L}_1,\ldots,\mathcal{L}_d$, namely
\begin{equation}
U[CC(\mathcal{L}) \bigotimes \mathcal{S}] U^{\dagger}= \bigotimes_{j=1}^{d} TC(\mathcal{L}_j),
\end{equation}
which concludes the proof of the (Equivalence) Theorem~\ref{th:equivalence}.

\section{Topological color code with boundaries}\label{sec:open}
\label{sec:boundaries}

Realistic physical systems have boundaries. Moreover, the transversal implementability of logical gates in the topological color code crucially depends on the choice of boundaries. In this section we show that the color code defined on a $d$-dimensional lattice with $d+1$ boundaries of $d+1$ distinct colors is equivalent to $d$ copies of the toric code attached together at a $(d-1)$-dimensional boundary. We also briefly describe how the choice of boundaries of the color code determines if the copes of the toric code are attached or decoupled. We then discuss such boundaries from the viewpoint of condensation of excitations.

\subsection{Physical intuition behind folding}
 
We begin with presenting some physical intuition why the toric code with two smooth and two rough boundaries needs to be folded if one hopes for transversal non-Pauli logical gates such as the Hadamard gate $\overline{H}$. Let us recall known results about gapped boundaries of the toric code. In two spatial dimensions, the toric code may have two types of boundaries, \emph{smooth} and \emph{rough}~\cite{Bravyi98}. The rough boundaries are defined as the boundaries with open edges (see Fig.~\ref{fig_fold_logical}). Similarly to the toric code without boundaries, there are $X$-vertex and $Z$-face stabilizers, although $Z$-face stabilizers have to be modified along the rough boundaries. An $X$-type  ($Z$-type) string-like logical operator can only start from and end on smooth (rough) boundaries. One says that the electric charge $e$, i.e. the violated $X$-vertex stabilizer, condenses on the rough boundary and the magnetic flux $m$, i.e. the violated $Z$-face stabilizer, is confined since single $e$, unlike $m$, can be created or absorbed on the rough boundary. Similarly, $m$ condenses and $e$ is confined on the smooth boundary.

\begin{figure}[h!]
\includegraphics[width=0.75\textwidth]{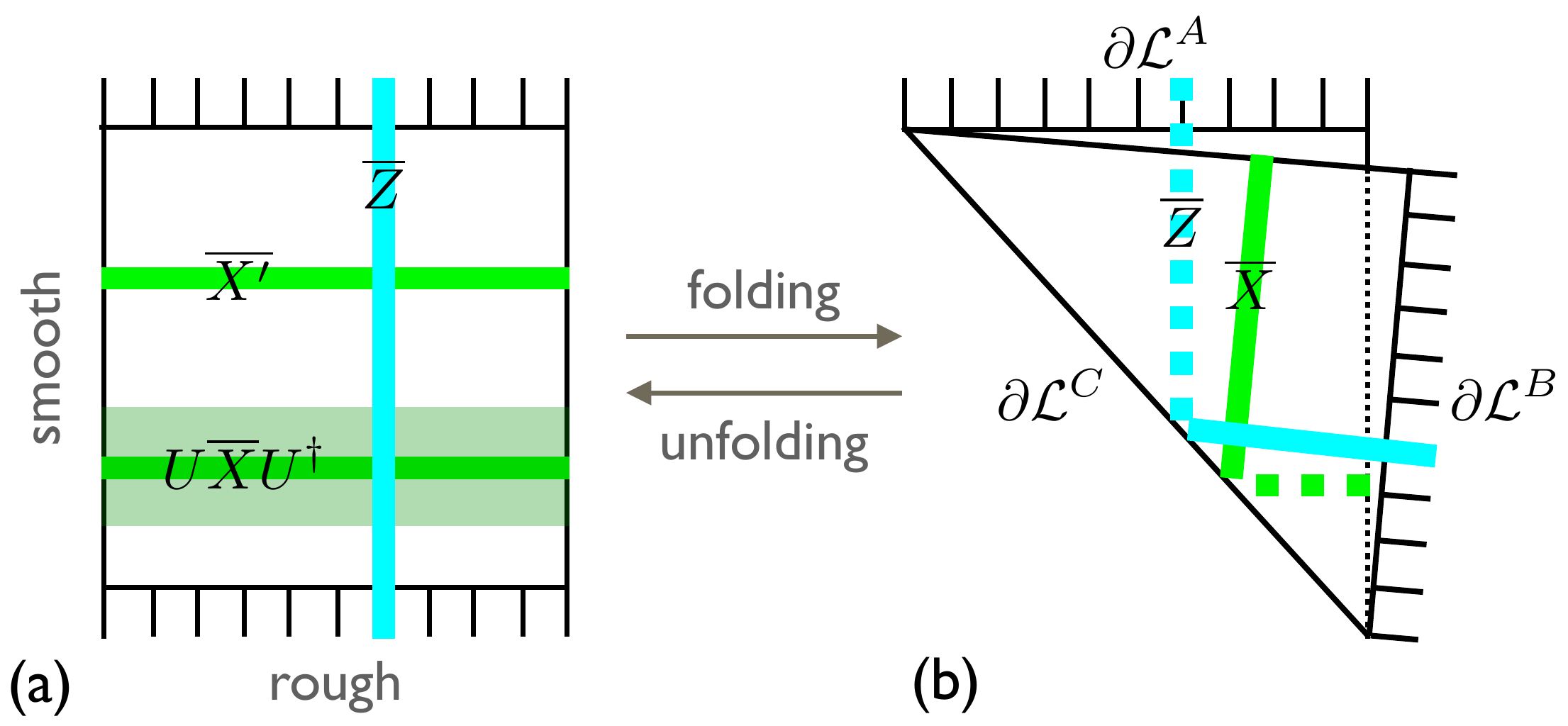}
\caption{(Color online) Origami of the toric code with boundaries. (a) Blue line, starting from and ending on rough boundaries, represents the logical $Z$ operator. Green lines, starting from and ending on smooth boundaries, represent the logical $X$ operator. (b) The color code with three boundaries,  $\partial\mathcal{L}^A$, $\partial\mathcal{L}^B$ and $\partial\mathcal{L}^C$, obtained by folding the toric code with two smooth and two rough boundaries. After folding, two logical operators $\overline{X}$ and $\overline{Z}$ are supported on overlapping regions.}
\label{fig_fold_logical} 
\end{figure} 

Consider the two-dimensional toric code with two smooth and two rough boundaries as depicted in Fig.~\ref{fig_fold_logical}(a). Since there is only one pair of anti-commuting logical operators, $\overline{X}$ and $\overline{Z}$, the code encodes a single logical qubit. There is one crucial difference between the toric code and the color code (with boundaries) --- the latter admits transversal implementation of the Hadamard gate $\overline{H}$ while the former does not. Recall that the Hadamard gate swaps Pauli $X$ and $Z$ operators. Suppose that the Hadamard gate can be implemented by a local unitary operator $U$. Let $\overline{X}$ and $\overline{X'}$ be two equivalent implementations of the logical $X$ operator, supported on string-like horizontal regions (see Fig.~\ref{fig_fold_logical}(a)). Then, $U\overline{X}U^{\dagger}$ implements the logical $Z$ operator, which has to anti-commute with $\overline{X'}$. On the other hand, since $U$ is a local unitary, then $U\overline{X}U^{\dagger}$ and $\overline{X'}$ have no overlap, and thus they commute, leading to a contradiction. We conclude that the logical Hadamard gate cannot be implemented by a local unitary operator in the toric code with boundaries. This is a simple version of the argument presented in Ref.~\cite{Beverland14}.

We note that if the logical Hadamard is transversal, then both logical $X$ and $Z$ operators must have representations which are supported on overlapping regions. By folding the toric code, both logical $X$ and $Z$ operators can be supported on overlapping regions, as shown in Fig.~\ref{fig_fold_logical}(b). Thus, for the logical Hadamard to be transversal folding of the toric code is indeed necessary.

\subsection{Unfolding in two dimensions}

We now return to the analysis of the topological color code $CC(\mathcal{L})$ supported on a ($3$-valent and $3$-colorable) two-dimensional lattice $\mathcal{L}$ with the Euler characteristic\footnote{We can think of $\mathcal{L}$ as a tiling of a $2$-manifold $\mathcal{M}$ with boundary, and then the Euler characteristic is $\chi = 2 -2g - b$, where $g$ is the genus of $\mathcal{M}$ and $b$ is the number of connected components of $\partial\mathcal{M}$.} $\chi$ and the boundary $\partial\mathcal{L}=\bigsqcup_{i=1}^n \partial\mathcal{L}^i$, where $\partial\mathcal{L}^i$ is the (maximum) connected component of the boundary $\partial\mathcal{L}$ of certain color. For conciseness, we simply refer to $\partial\mathcal{L}^i$ as a boundary. We say that the boundary $\partial\mathcal{L}^i$ is of color $C_1$ if all the faces adjacent to $\partial\mathcal{L}^i$ have colors $C_2$ and $C_3$, where $\{ C_1, C_2, C_3 \}=\{A,B,C \}$. One can show that the color code $CC(\mathcal{L})$ encodes $n-2\chi$ logical qubits. In particular, one important case corresponds to the triangular color code (with three boundaries of color $A$, $B$ and $C$ as shown in Fig.~\ref{fig_2dim_boundary}(a); see also Fig.~\ref{fig_fold_logical}(b)), which encodes one logical qubit regardless of the system size, and has transversal logical Hadamard $\overline{H}$ and the phase gate $\overline{R_{2}}$ .

We would like to understand how the color code $CC(\mathcal{L})$ with boundaries transforms under the disentangling unitary $U=\bigotimes_{f\in\mathcal{C}} U_f$ described in Section~\ref{sec:closed}B. In the bulk, the disentangling unitary $U$ transforms the stabilizers of the color code into stabilizers of the toric code supported on two decoupled lattices $\mathcal{L}_A$ and $\mathcal{L}_B$, obtained from $\mathcal{L}$ by shrinking faces of color $A$ and $B$, respectively. On the other hand, the stabilizers of the color code supported on qubits near the boundaries may transform into stabilizers supported on both shrunk lattices $\mathcal{L}_{A}$ and $\mathcal{L}_{B}$, depending on the colors of $\partial\mathcal{L}$. In general, we cannot transform the color code $CC(\mathcal{L})$ into the toric code supported on two decoupled lattices, $TC(\mathcal{L}_{A})\otimes TC(\mathcal{L}_{B})$. Rather, the toric code is defined on a lattice $\mathcal{L}_{A}\# \mathcal{L}_{B}$ obtained by \emph{attaching}\footnote{We would like to point out similarities between the attaching procedure we describe and welding defined in Ref.~\cite{Michnicki2014}.} $\mathcal{L}_{A}$ and $\mathcal{L}_{B}$, i.e. identifying some of their boundaries. Namely,
\begin{equation}
U[CC(\mathcal{L})] U^{\dagger}= TC(\mathcal{L}_{A}\# \mathcal{L}_{B}).
\end{equation}
In the rest of this subsection we analyze the triangular color code (see Fig.~\ref{fig_2dim_boundary}), but the discussion is applicable to the color code on any homogeneous cell $2$-complex with boundary, which is $3$-colorable and $3$-valent.

\begin{figure}[h!]
\includegraphics[width=1.00\textwidth]{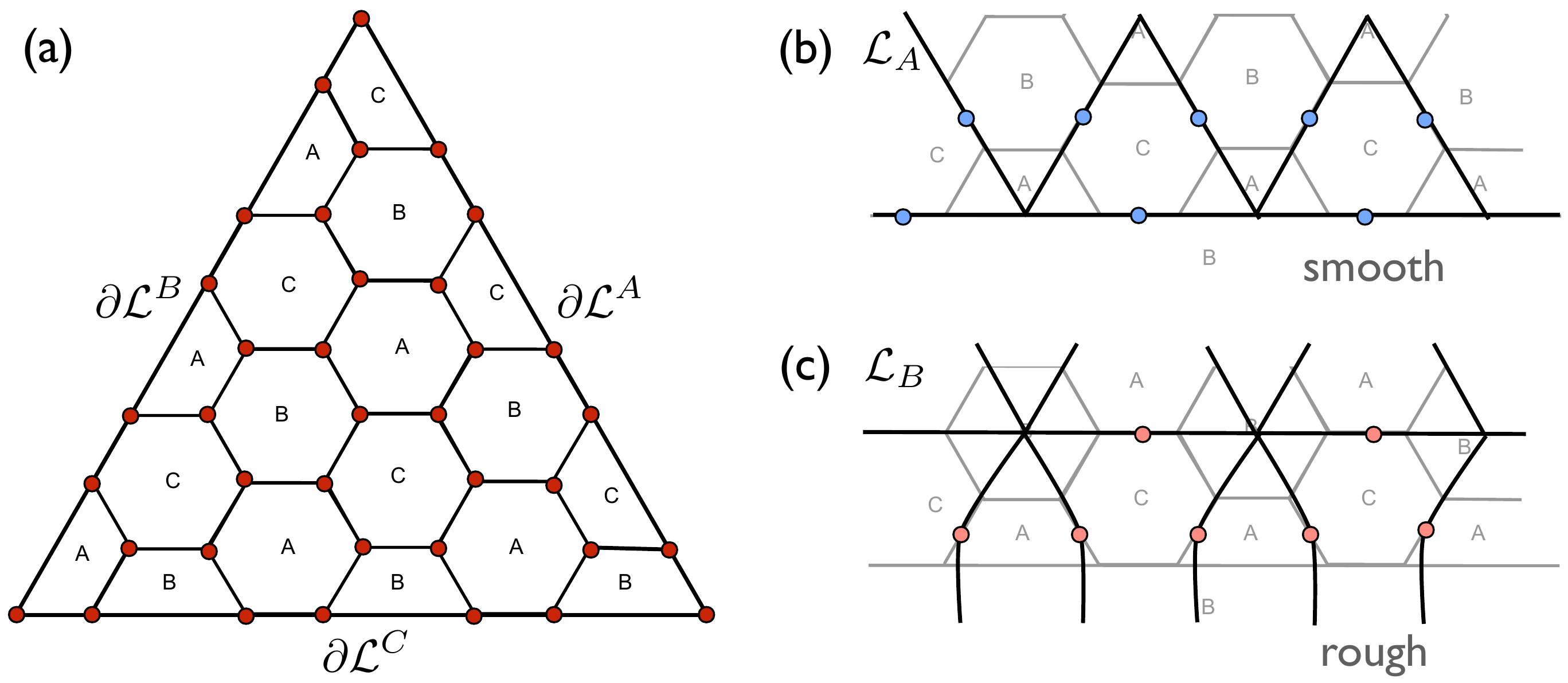}
\caption{(Color online) (a) The (triangular) color code on a two-dimensional lattice $\mathcal{L}$ with the boundary $\partial\mathcal{L}$ comprising of three components of color $A$, $B$ and $C$, namely $\partial\mathcal{L}=\partial\mathcal{L}^A\sqcup\partial\mathcal{L}^B\sqcup \partial\mathcal{L}^C$. Qubits are represented by dots. (b) A fragment of the lattice $\mathcal{L}_{A}$ derived from $\mathcal{L}$ by shrinking faces of color $A$. The smooth boundary arises in $\mathcal{L}_{A}$ on the boundary $\partial\mathcal{L}^B$. (c) A fragment of the lattice $\mathcal{L}_{B}$ derived from $\mathcal{L}$ by shrinking faces of color $B$. The rough boundary arises in $\mathcal{L}_{B}$ on the boundary $\partial\mathcal{L}^B$ . }
\label{fig_2dim_boundary} 
\end{figure}

Let us describe how to obtain the lattice $\mathcal{L}_{A}\# \mathcal{L}_{B}$ supporting the toric code. Recall that in the bulk, $\mathcal{L}_{A}$ and $\mathcal{L}_{B}$ are obtained from $\mathcal{L}$ by shrinking faces of color $A$ and $B$. Let $\partial\mathcal{L}^A$, $\partial\mathcal{L}^B$ and $\partial\mathcal{L}^C$ be the boundaries of color $A$, $B$ and $C$, respectively. We find that shrunk lattices $\mathcal{L}_{A}$ and $\mathcal{L}_{B}$ are decoupled along $\partial\mathcal{L}^A$ and $\partial\mathcal{L}^B$, but are identified along $\partial\mathcal{L}^C$. In particular,
\begin{itemize}
\item  on the boundary $\partial\mathcal{L}^A$: the lattice $\mathcal{L}_{B}$ has open edges (rough boundary), whereas $\mathcal{L}_{A}$ --- no open edges (smooth boundary),
\item  on the boundary $\partial\mathcal{L}^B$: the lattice $\mathcal{L}_{A}$ has open edges (rough boundary), whereas $\mathcal{L}_{B}$ --- no open edges (smooth boundary),
\item  on the boundary $\partial\mathcal{L}^C$: since the disentangling unitary $U$ does not affect the qubits placed on vertices belonging to $\partial\mathcal{L}^C$, both lattices $\mathcal{L}_{A}$ and $\mathcal{L}_{B}$ share these qubits.
\end{itemize}
See Fig.~\ref{fig_2dim_boundary} and Fig.~\ref{fig_weld}(a)(b) for an example of how smooth and rough boundaries arise in the disentangling procedure. Note that on $\partial\mathcal{L}^C$, the lattices $\mathcal{L}_{A}$ and $\mathcal{L}_{B}$ are identified. This implies that an $e$ excitation on $\mathcal{L}_{A}$ can be transformed into an $e$ excitation on $\mathcal{L}_{B}$ by going through the boundary $\partial\mathcal{L}^C$; similarly for $m$ excitations.

We can visualize the lattice $\mathcal{L}_{A}\# \mathcal{L}_{B}$ by flipping vertically $\mathcal{L}_{B}$ and attaching it to $\mathcal{L}_{A}$ (see Fig.~\ref{fig_weld}(c)). Observe that starting from the color code $CC(\mathcal{L})$ with three boundaries, performing the disentangling unitary $U=\bigotimes_{f\in\mathcal{C}} U_f$ and unfolding the resulting lattice $\mathcal{L}_{A}\# \mathcal{L}_{B}$, one obtains a single copy of the toric code $TC(\mathcal{L}_{A}\# \mathcal{L}_{B})$ with two smooth and two rough boundaries. We can summarize the discussion by the following theorem.

\begin{theorem}[Unfolding]
The (triangular) color code $CC(\mathcal{L})$ on a two-dimensional lattice $\mathcal{L}$ with three boundaries, $\partial\mathcal{L}^A$, $\partial\mathcal{L}^B$ and $\partial\mathcal{L}^C$, is equivalent to one folded copy of the toric code $TC(\mathcal{L}_{A}\# \mathcal{L}_{B})$ defined on a lattice $\mathcal{L}_{A}\# \mathcal{L}_{B}$ with two smooth and two rough boundaries. Moreover, $\mathcal{L}_{A}\# \mathcal{L}_{B}$ is constructed by attaching two lattices $\mathcal{L}_{A}$ and $\mathcal{L}_{B}$ (derived from $\mathcal{L}$ by shrinking faces of color $A$ and $B$, respectively) along the boundary $\partial\mathcal{L}^C$.
\end{theorem}

\begin{figure}[h!]
\includegraphics[width=\textwidth]{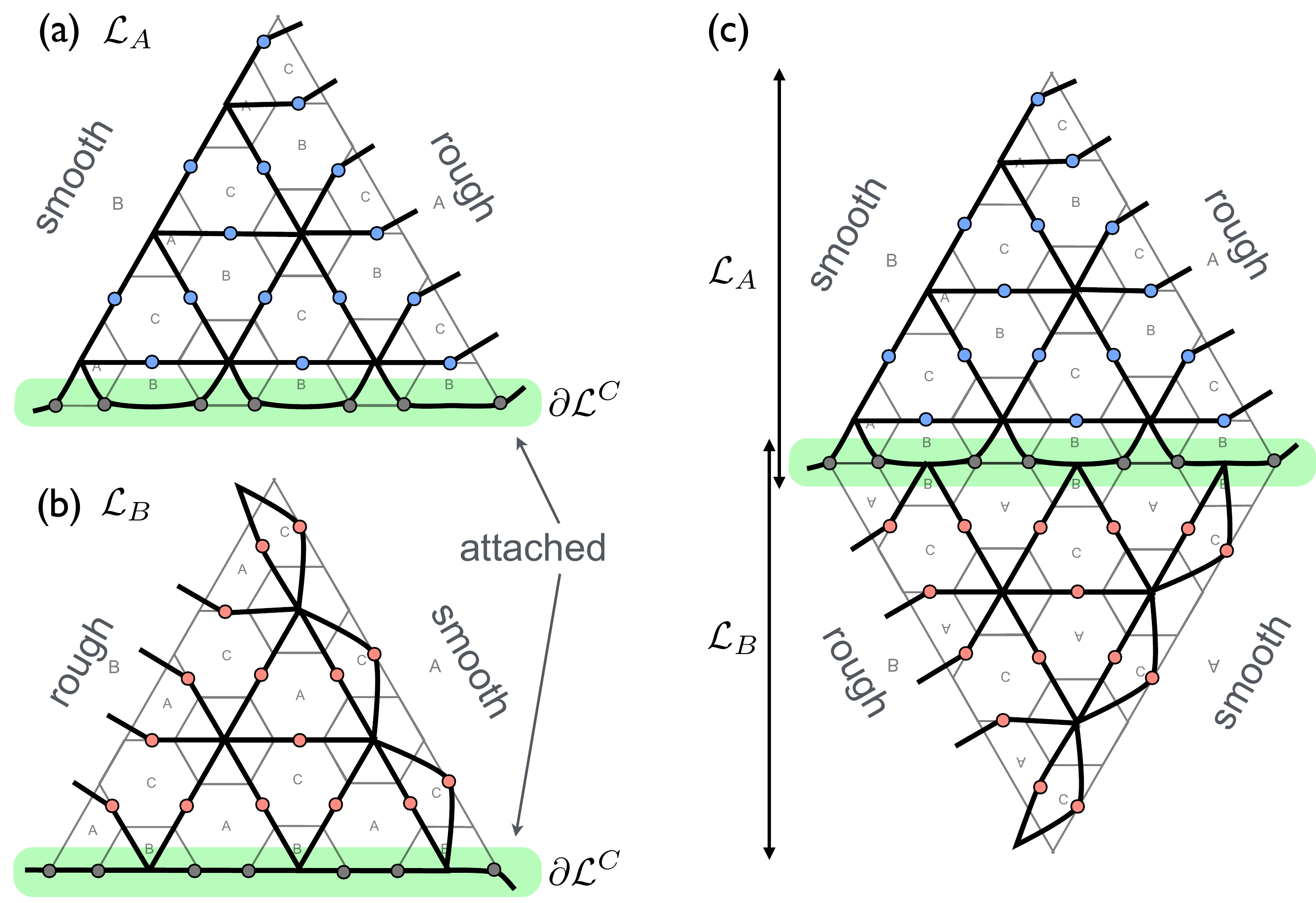}
\caption{(Color online) Attaching two lattices: (a) $\mathcal{L}_{A}$ and (b) $\mathcal{L}_{B}$ by identifying qubits along the boundary $\partial\mathcal{L}^C$. Note that both $\mathcal{L}_{A}$ and $\mathcal{L}_{B}$ have two qubits per edge on the boundary $\partial\mathcal{L}^C$. (c) Unfolded toric code $TC(\mathcal{L}_A\#\mathcal{L}_B)$. Blue qubits belong to the lattice $\mathcal{L}_{A}$, whereas red qubits belong to the (flipped) lattice $\mathcal{L}_{B}$.}
\label{fig_weld} 
\end{figure}

\subsection{Three (or more) dimensions}

The toric code on a $d$-dimensional lattice with boundaries, $d\geq 3$, does not differ substantially from the two-dimensional model --- qubits are placed on edges, and $X$- and $Z$-type stabilizer generators are associated with vertices and faces. There are two types of boundaries, rough and smooth, which may absorb point-like electric charges and $(d-1)$-dimensional magnetic fluxes, respectively. Moreover, string-like logical $Z$ (respectively $(d-1)$-dimensional membrane-like logical $X$) operators can only start from and end on rough (respectively smooth) boundaries (see Fig.~\ref{fig_logical_boundary}(b)).

The color code can be defined on a $(d+1)$-valent and $(d+1)$-colorable $d$-dimensional lattice $\mathcal{L}$ with the boundaries $\partial\mathcal{L}=\bigsqcup_{i=1}^n \partial\mathcal{L}^i$, where each (maximum) connected component $\partial\mathcal{L}^i$ has one out of $d+1$ colors, $C_0,\ldots,C_d$. We say that $\partial\mathcal{L}^i$ is of color $C_j$ if all $d$-cells adjacent to $\partial\mathcal{L}^i$ have colors different from $C_j$. Qubits are placed on vertices, and $X$- and $Z$-type stabilizer generators are associated with $d$-cells and faces, respectively. For the sake of clarity, in the rest of this subsection we focus on the three-dimensional color code $CC(\mathcal{L})$ defined on a tetrahedron-like lattice $\mathcal{L}$ with four boundaries of color $A$, $B$, $C$ and $D$ (see Fig.~\ref{fig_3D_boundary}(a)).

\begin{figure}[h!]
\includegraphics[width=\textwidth]{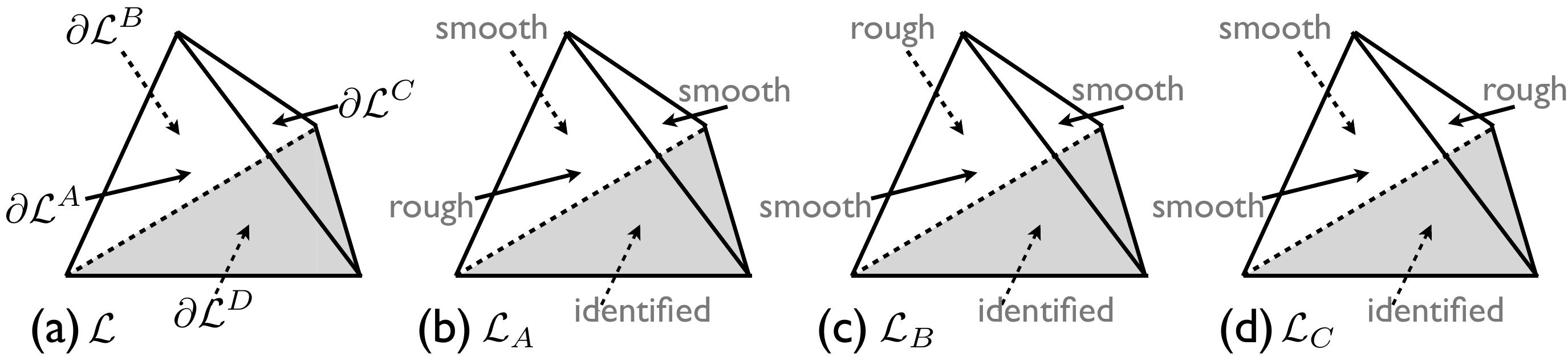}
\caption{ (a) A tetrahedron-like lattice $\mathcal{L}$ with boundaries $\partial\mathcal{L}^{A}$, $\partial\mathcal{L}^{B}$, $\partial\mathcal{L}^C$ and $\partial\mathcal{L}^D$. Three shrunk lattices: (b) $\mathcal{L}_{A}$, (c) $\mathcal{L}_{B}$, (d) $\mathcal{L}_{C}$ derived from $\mathcal{L}$. The shaded boundary represents the attaching boundary $\partial\mathcal{L}^D$.}
\label{fig_3D_boundary} 
\end{figure} 

We would like to analyze what happens to $CC(\mathcal{L})$ if we apply the disentangling unitary $U=\bigotimes_{c\in\mathcal{D}} U_c$ described in Section~\ref{sec:closed}D. In the bulk, the disentangling unitary $U$ transforms the stabilizers of the color code into stabilizers of the toric code supported on three decoupled lattices, $\mathcal{L}_A$, $\mathcal{L}_B$ and $\mathcal{L}_C$, obtained from $\mathcal{L}$ by shrinking volumes of color $A$, $B$ and $C$, respectively. 
Since $\mathcal{L}_A$, $\mathcal{L}_B$ and $\mathcal{L}_C$ share qubits along the boundary $\partial\mathcal{L}^D$, we cannot transform the color code $CC(\mathcal{L})$ into the toric code supported on three decoupled lattices, $TC(\mathcal{L}_{A})\otimes TC(\mathcal{L}_{B})\otimes TC(\mathcal{L}_{C})$. Rather, the toric code is defined on a lattice $\mathcal{L}_{A}\# \mathcal{L}_{B}\# \mathcal{L}_{C}$ obtained by attaching three shrunk lattices along the boundary $\partial\mathcal{L}_D$. We then obtain
\begin{equation}
U[CC(\mathcal{L})\otimes\mathcal{S}] U^{\dagger}= TC(\mathcal{L}_{A}\# \mathcal{L}_{B}\# \mathcal{L}_{C}),
\end{equation}
where $\mathcal{S}$ is the stabilizer group of the ancilla qubits. Note that $U$ does not transform qubits on vertices belonging to the boundary $\partial\mathcal{L}_D$.

Let us have a closer look at the shrunk lattices and the identified boundary. $\mathcal{L}_i$ has one rough boundary $\partial\mathcal{L}^i$, and two smooth boundaries $\partial\mathcal{L}^j$ and $\partial\mathcal{L}^k$, where $\{i,j,k\}=\{A,B,C\}$ (see Fig.~\ref{fig_3D_boundary}). Recall that shrunk lattices share only qubits placed on vertices of the identified boundary $\partial\mathcal{L}_D$. To obtain $\mathcal{L}_{A}\# \mathcal{L}_{B}\# \mathcal{L}_{C}$ one attaches the shrunk lattices by identifying the qubits placed on vertices of $\partial\mathcal{L}^D$ (see Fig.~\ref{fig_attach}). Note that a single-qubit Pauli $Z$ operator on a qubit on the boundary $\partial\mathcal{L}^D$ causes three $X$-vertex stabilizers to be violated, i.e. one in each of three shrunk lattices. Put another way, such an operator creates a triple of electric charges, $e_{A}$, $e_{B}$ and $e_{C}$. This implies that the composite electric charge $e_{A}e_{B}e_{C}$ can condense on the boundary $\partial\mathcal{L}^D$. We focus on condensation of excitations on the boundaries in the next subsection.

\begin{figure}[h!]
\includegraphics[width=\textwidth]{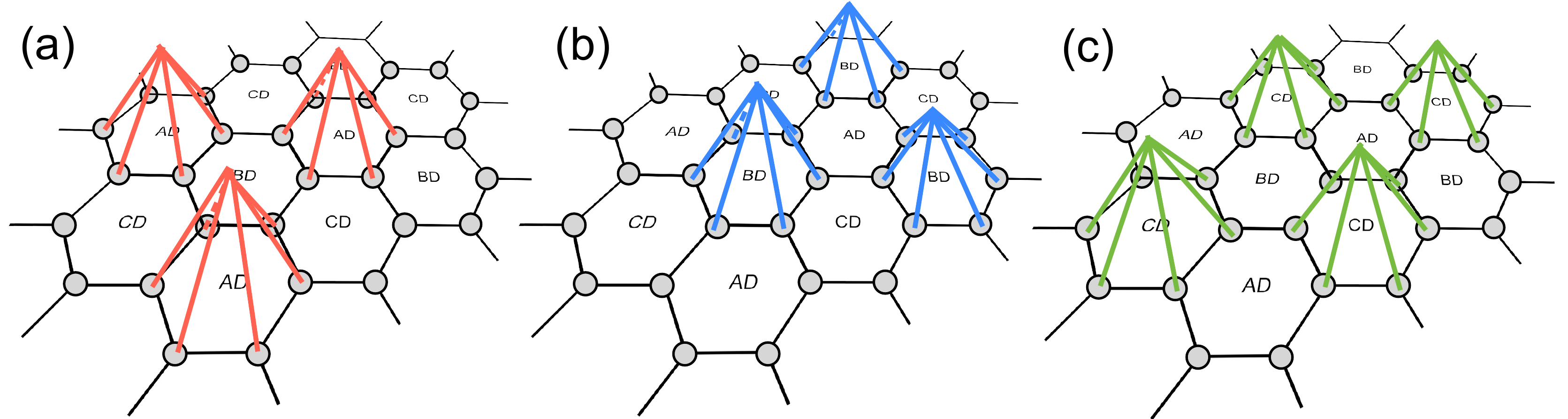}
\caption{(Color online) The identified boundary $\partial\mathcal{L}^D$ of three shrunk lattices: (a) $\mathcal{L}_{A}$, (b) $\mathcal{L}_{B}$ and (c) $\mathcal{L}_{C}$. The shrunk lattices are attached by identifying qubits on vertices of $\partial\mathcal{L}^{D}$.}
\label{fig_attach} 
\end{figure} 

The discussion here can be straightforwardly generalized to $d$ dimensions, yielding the equivalence between the color code and the toric code with boundaries. We conclude with the following theorem.

\begin{theorem}[Attaching]
Let $CC(\mathcal{L})$  be the color code on a $d$-simplex-like lattice $\mathcal{L}$ with $d+1$ boundaries $\partial\mathcal{L}^0,\ldots,\partial\mathcal{L}^d$, where $\partial\mathcal{L}^i$ has color $C_i$. Then, there exists a local Clifford unitary $U=\bigotimes_{c\in\mathcal{C}_0} U_c$ (described in Section~\ref{sec:closed}D) such that
\begin{equation}
U[ CC(\mathcal{L})\otimes\mathcal{S} ]U^\dag = TC(\#_{i=1}^d \mathcal{L}_i),
\end{equation}
where $\mathcal{S}$ is the stabilizer group of the ancilla qubits. The toric code $TC(\#_{i=1}^d \mathcal{L}_i)$ is defined on the lattice 
$\#_{i=1}^d \mathcal{L}_i$ obtained by attaching lattices $\mathcal{L}_1,\ldots,\mathcal{L}_d$ along the boundary $\partial\mathcal{L}^0$, where $\mathcal{L}_i$ is derived from $\mathcal{L}$ by shrinking $d$-cells of color $C_i$, and has one rough boundary, $\partial\mathcal{L}_i$.
\end{theorem}

\subsection{Condensation of anyonic excitations}

It is instructive to interpret the equivalence between the color code and the toric code with boundaries from the viewpoint of condensation of anyonic excitations. In the two-dimensional toric code, the anyonic excitations are: electric $e$ --- a single violated $X$-vertex stabilizer, magnetic $m$ --- a single violated $Z$-face stabilizer, and fermionic $\epsilon = e \times m$  --- a composite excitation obtained by fusing $e$ and $m$. The label $1$ corresponds to the vacuum (no excitations).

The gapped boundaries of two-dimensional systems are classified by maximum sets of mutually bosonic excitations which may condense~\cite{Kitaev12,Levin12,Lan14}. In the case of a single layer of the toric code, possible sets of anyons which may condense on the boundaries are $\{ 1, e \}$ and $\{ 1, m \}$. Note that $\epsilon$ has fermionic self-statistics and thus cannot condense on the gapped boundaries. The sets $\{ 1, e \}$ and $\{ 1, m \}$ correspond to rough and smooth boundaries, respectively~\cite{Bravyi98}. On the other hand, the folded toric code has three boundaries (see Fig.~\ref{fig_fold_logical}(b)). If we denote by $e_i$, $m_i$ and $\epsilon_i$ the excitations in the front ($i=1$) and rear ($i=2$) layer of the folded toric code, then we can associate the boundaries with the sets of condensing anyons. Namely,
\begin{eqnarray}
\partial\mathcal{L}^A &\leftrightarrow& \{ 1, e_{1}, m_{2}, e_{1}m_{2} \}\\
\partial\mathcal{L}^B &\leftrightarrow& \{ 1, e_{2}, m_{1}, e_{2}m_{1} \}\\ 
\partial\mathcal{L}^C &\leftrightarrow& \{ 1, e_{1}e_{2}, m_{1}m_{2}, \epsilon_{1}\epsilon_{2} \}
\end{eqnarray}
As depicted in Fig.~\ref{fig_condensation}(a), two electric charges $e_{1}$ and $e_{2}$ created on boundaries $\partial\mathcal{L}^A$ and $\partial\mathcal{L}^B$ can be jointly annihilated (or created) on $\partial\mathcal{L}^C$. 

\begin{figure}[h!]
\includegraphics[width=0.7\textwidth]{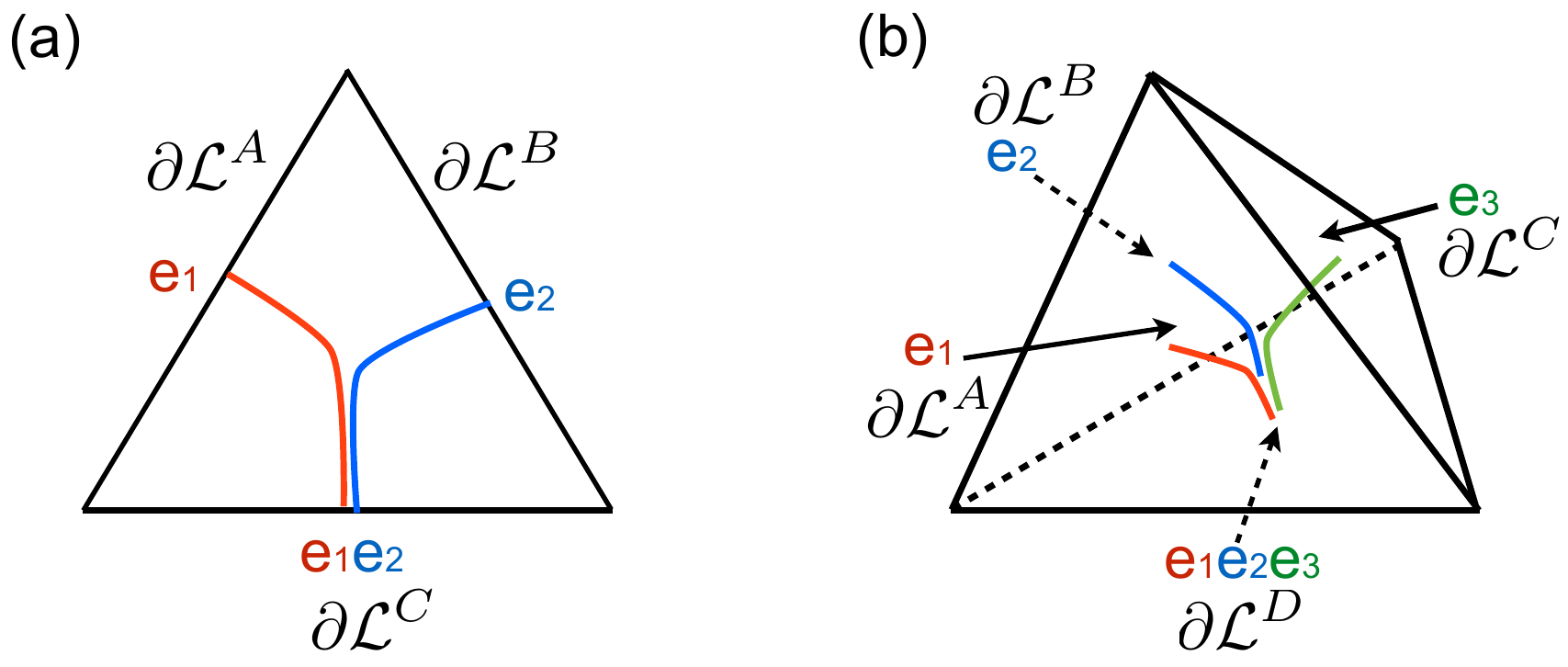}
\caption{(Color online) Condensation of electric charges in (a) two and (b) three dimensions. Observe that single electric charges can condense on all but one boundary, which is the identified boundary. On the identified boundary, a composite electric charge (a) $e_1 e_2$ and (b) $e_1 e_2 e_3$ can be created or annihilated. }
\label{fig_condensation} 
\end{figure}

By associating the boundaries with the sets of condensing anyons we can find the correspondence between anyonic excitations in the toric code and the color code. We can label excitations in the color code by $i_P$, where $P \in \{ X, Z \}$ indicates the type of the violated stabilizer, and $i \in \{ A, B, C \}$ indicates the color of the face associated with the violated stabilizer. Observe that not all six excitations are independent. For instance, a single qubit Pauli $X$ operator on a vertex $v$ creates excitations, $A_{Z}$, $B_{Z}$ and $C_{Z}$, on three neighboring faces sharing $v$. This implies that in the bulk the following fusion channels exist
\begin{equation}
A_{X}\times B_{X}\times C_{X} = 1,\qquad A_{Z}\times B_{Z}\times C_{Z} = 1. 
\end{equation}
Note that excitations $i_{X}$ and $i_{Z}$ can only condense on the boundary $\partial\mathcal{L}^i$. This leads to the following isomorphism between labels of anyonic excitations of the toric and color codes
\begin{align}
e_{1} \leftrightarrow A_{X}, \qquad e_{2} \leftrightarrow B_{X}, \qquad  m_{1} \leftrightarrow B_{Z}, \qquad m_{2} \leftrightarrow A_{Z}. 
\end{align}

In $d>2$ dimensions, the excitations of the color code are point-like electric charges and $(d-1)$-dimensional magnetic fluxes. Let us first focus on condensation of electric charges. We find that the boundaries of the $d$-dimensional color code on a $d$-simplex-like lattice are given by
\begin{eqnarray}
\partial\mathcal{L}^0 &\leftrightarrow& \{ e_{1}e_{2}\ldots e_{d} \},\\
\partial\mathcal{L}^i &\leftrightarrow& \{  e_{i} \} \qquad \mathrm{for}\ i=1,\ldots,d.
\end{eqnarray}
(See Fig.~\ref{fig_condensation} for two- and three-dimensional examples). Yet, none of the magnetic fluxes can individually condense on the boundary $\partial\mathcal{L}^0$.  Rather, any pair of fluxes can condense on $\partial\mathcal{L}^0$, and thus we might think of the fluxes as being equivalent. To sum up, we find the following condensations of $(d-1)$-dimensional magnetic fluxes:
\begin{eqnarray}
\partial\mathcal{L}^0 &\leftrightarrow& \left\{ m_{i}m_{j} |\ \forall i\neq j \right\},\\
\partial\mathcal{L}^i &\leftrightarrow& \left\{ m_{j} |\ \forall j\neq i \right\}.
\end{eqnarray}
One may observe that, as expected, the set of condensing magnetic and electric excitations on every boundary is mutually bosonic.

We would like to emphasize that while the gapped boundaries in $(2+1)$-dimensional TQFTs have been throughly classified \cite{Kitaev12,Beigi11}, the understanding of the gapped boundaries in higher-dimensional TQFTs is still incomplete. Characterization of condensing anyonic excitations in the color code may provide instructive examples helping with classification of the gapped boundaries in higher-dimensional TQFTs. Namely, different boundaries of various colors in the color code may lead to a rich variety of gapped boundaries in the corresponding toric code models. Moreover, logical action of the transversal $\widetilde{R}_n$ operator on the code space crucially depends on the choice of boundaries in the color code. Thus, one may be able to characterize gapped boundaries by analyzing the logical action of transversal operators, and vice versa.

\section{Transversal gates}
\label{sec:gate}

We have seen that the color code is equivalent to (multiple copies of) the toric code, both in the presence or the absence of boundaries. Our findings hint that there might be non-trivial logical gates from the $d$-th level of the Clifford hierarchy in the $d$-dimensional toric code which admit fault-tolerant implementation. In this section, we show that one can implement by local unitary transformations the logical $d$-qubit control-$Z$ gate on the stack of $d$ copies of the $d$-dimensional toric code with point-like excitations.

\subsection{Transversal $R_{d}$ operator and boundaries}

Let us start with reviewing the transversal implementation of the physical phase gate $R_{n}=\text{diag}(1,e^{2 \pi i / 2^n})$ in the color code \cite{Bombin13,Kubica15}. Consider the topological color code $CC(\mathcal{L})$ on a $d$-dimensional lattice $\mathcal{L}$, which is $(d+1)$-valent and $(d+1)$- colorable. It is known that the graph $G=(V,E)$ of vertices and edges of $\mathcal{L}$ is bipartite, namely the set of vertices $V$ can be split into two subsets, $T$ and $T^{c}$, such that $V=T\sqcup T^{c}$ and vertices in $T$ are connected only to vertices in $T^{c}$, and vice versa. Then, regardless of the lattice $\mathcal{L}$, the following unitary operator preserves the code space
\begin{equation}
\label{eq:Rtransversal}
\widetilde{R_{d}}= \bigotimes_{j \in T} R_{d}(j) \bigotimes_{j \in T^{c}} R_{d}^{-1}(j).
\end{equation}
Here, we adopt a convention that $\widetilde{R_{d}}$ denotes a transversal operator implemented by physical $R_{d}$ gates or their powers. When the lattice $\mathcal{L}$ is $d$-simplex-like (see Section~\ref{sec:boundaries}C and Fig.~\ref{fig_3D_boundary}), then $\widetilde{R_{d}}$ implements the logical $R_{d}$ gate in the code space. For other choices of boundaries, the action of $\widetilde{R_{d}}$ in the code space does not necessarily coincide with the logical $R_{d}$ gate.

For the sake of simplicity, in the rest of this section we shall consider the $d$-dimensional color code supported on a $d$-hypercube-like lattice $\mathcal{L}$ colored with $C_0,\ldots,C_d$  (see Fig.~\ref{fig_squares}(a) and Fig.~\ref{fig_logical_boundary}(a)). In particular, we choose $\mathcal{L}$ to have the opposite boundaries colored in the same color. Namely, we assume that two boundaries perpendicular to the direction $\hat{j}$ have color $C_{j}$, where $j=1,\ldots,d$. One can show that the color code $CC(\mathcal{L})$ encodes $d$ logical qubits. In order to do this, consider the disentangling unitary $U=\bigotimes_{c\in\mathcal{C}_0} U_c$, which is a tensor product of local unitaries supported on $d$-cells of color $C_0$ (see Section~\ref{sec:closed}D and Section~\ref{sec:open}C). Then, $U$ transforms the color code $CC(\mathcal{L})$ into $d$ decoupled copies of the toric code,
\begin{equation}
\label{eq:decouplingoncube}
U [ CC(\mathcal{L}) \otimes \mathcal{S}] U^\dag = \bigotimes_{i=1}^d TC(\mathcal{L}_{i})
\end{equation}
where $\mathcal{S}$ is the stabilizer group of the ancilla qubits and the lattice $\mathcal{L}_{i}$ is derived from $\mathcal{L}$ by shrinking $d$-cells of color $C_i$. Moreover, $\mathcal{L}_{i}$ is a $d$-hypercube-like lattice with two rough boundaries which are perpendicular to the direction $\hat{i}$ and all the other boundaries smooth. Thus, for $i=1,\ldots,d$, the toric code $TC(\mathcal{L}_i)$ encodes one logical qubit, with a string-like logical $Z$ operator in the direction $\hat{i}$, and a $(d-1)$-dimensional membrane-like logical $X$ operator perpendicular to $\hat{i}$.

With the above choice of boundaries, $\widetilde{R_{d}}$ does not implement the logical $R_{d}$ gate in the code space. One verifies this by observing that $\widetilde{R_{d}}^{2}=I$ in the code space of the color code. Rather, we find that $\widetilde{R_{d}}$ implements the logical $d$-qubit control-$Z$ gate on the stack of $d$ copies of the toric code. (Note that a similar observation holds for the color code supported on a hypercubic lattice with periodic boundary conditions). We devote the rest of this section to describe this finding.

\begin{figure}[h!]
\includegraphics[width=0.85\textwidth]{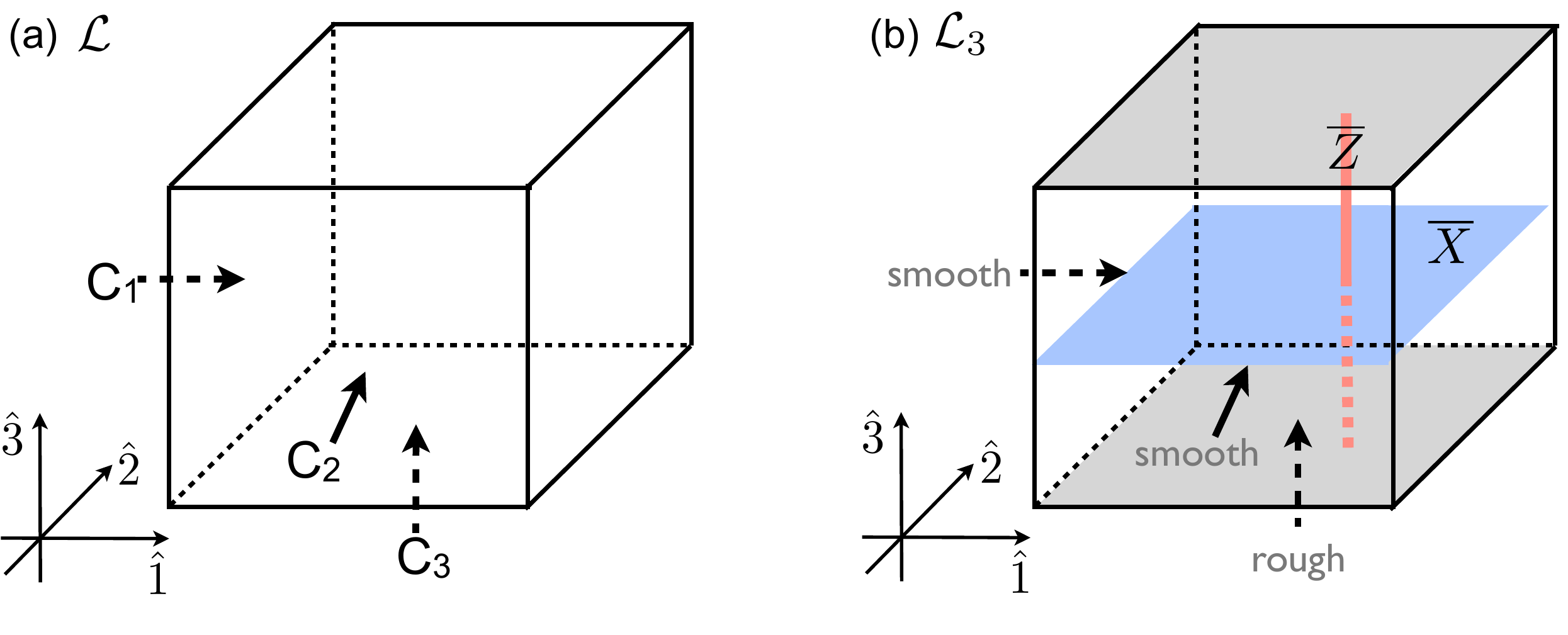}
\caption{(Color online) (a) The color code $CC(\mathcal{L})$ on a three-dimensional cube-like lattice $\mathcal{L}$ with pairs of boundaries perpendicular to the direction $\hat{i}$ colored with $C_i$. (b) The toric code $TC(\mathcal{L}_3)$ on a cube-like lattice $\mathcal{L}_3$ derived from $\mathcal{L}$ by shrinking $3$-cells of color $C_3$. Note that $\mathcal{L}_3$ has two rough boundaries (shaded) and $TC(\mathcal{L}_3)$ encodes one logical qubit with a string-like logical $Z$ operator (red) connecting two opposite rough boundaries and a membrane-like logical $X$ operator (blue).}
\label{fig_logical_boundary} 
\end{figure}

\subsection{Transversal $d$-qubit control-$Z$ gate in the toric code}

We discuss the two-dimensional case first. The topological color code on a square-like lattice $\mathcal{L}$ with four boundaries of color $C_1$ and $C_2$ encodes two logical qubits (see Fig.~\ref{fig_squares}(a)). We label by $\overline{X^{(i)}}$ and $\overline{Z^{(i)}}$ the logical Pauli $X$ and $Z$ operators, which perpendicular or parallel to the direction $\hat{i}$, respectively, for $i=1,2$. Since a unitary operator $R_{2}$ transforms Pauli $X$ into $XZ$, logical operators transform under conjugation by $\widetilde{R_{2}}$ as follows

\begin{equation}
\overline{X^{(1)}}\rightarrow \overline{X^{(1)}}\ \overline{Z^{(2)}},\qquad 
\overline{Z^{(1)}}\rightarrow \overline{Z^{(1)}},\qquad
\overline{X^{(2)}}\rightarrow \overline{Z^{(1)}}\ \overline{X^{(2)}},\qquad 
\overline{Z^{(2)}}\rightarrow \overline{Z^{(2)}}.
\end{equation}

\begin{figure}[h!]
\includegraphics[width=0.85\textwidth]{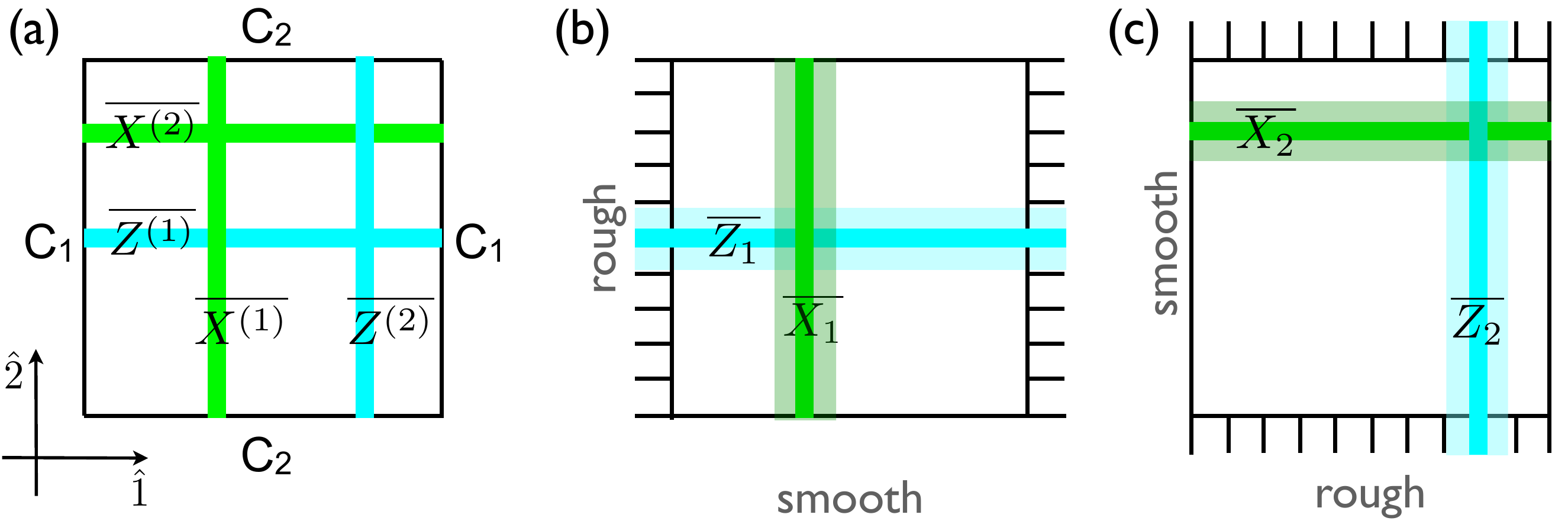}
\caption{(Color online) (a) The topological color code $CC(\mathcal{L})$ on a square-like lattice $\mathcal{L}$ with four boundaries of color $C_1$ and $C_2$ encodes two logical qubits, with logical operators $\overline{X^{(i)}}$ and $\overline{Z^{(i)}}$ for $i=1,2$. The toric code $TC(\mathcal{L}_i)$ (b) for $i=1$ and (c) for $i=2$ derived from $\mathcal{L}$ by shrinking faces of color $C_i$ encodes one logical qubit with logical operators  $\overline{X_i}$ and $\overline{Z_i}$.}
\label{fig_squares} 
\end{figure}

Note that the disentangling unitary $U$ (see Eq.~(\ref{eq:decouplingoncube})) transforming the color code $CC(\mathcal{L})$ into two decoupled copies of the toric code, $TC(\mathcal{L}_1)$ and $TC(\mathcal{L}_2)$, defines an isomorphism between logical operators of the former and the latter (see Fig.~\ref{fig_squares}). Namely,
\begin{equation}
\overline{X^{(1)}} \leftrightarrow  \overline{X_1}\otimes I, \qquad 
\overline{X^{(2)}} \leftrightarrow  I \otimes \overline{X_2}, \qquad 
\overline{Z^{(1)}}\ \leftrightarrow  \overline{Z_1} \otimes I, \qquad 
\overline{Z^{(2)}} \leftrightarrow   I \otimes \overline{Z_2},
\end{equation}
where $\overline{P_1}\otimes\overline{P_2}$ denotes an operator which acts as a logical $P_1$ operator on the first copy $TC(\mathcal{L}_1)$ of the toric code, and as $P_2$ on the second copy $TC(\mathcal{L}_2)$. Thus, one can immediately deduce the effect of $U\widetilde{R_{2}} U^\dag$ on logical operators of $TC(\mathcal{L}_1)$ and $TC(\mathcal{L}_2)$
\begin{equation}
\overline{X_1}\otimes I \rightarrow \overline{X_1}\otimes \overline{Z_2}, \quad
I\otimes \overline{X_2} \rightarrow \overline{Z_1}\otimes \overline{X_2}, \quad
\overline{Z_1}\otimes I \rightarrow\overline{Z_1}\otimes I, \quad
I \otimes \overline{Z_2} \rightarrow I \otimes \overline{Z_2}.
\end{equation}
This implies that the action of $\widetilde{R_{2}}$ in the color code is equivalent (up to the local Clifford unitary $U$) to the logical control-$Z$ gate between two copies of the toric code. 

Let us consider $d$-dimensional case, $d\geq 2$. The $d$-qubit control-$Z$ gate is a generalization of the control-$Z$ gate and is defined in the computational basis as
\begin{align}
\mbox{C}^{\otimes d-1}Z | x_{1},\ldots,x_{d} \rangle= (-1)^{x_{1}\ldots x_{d}}| x_{1},\ldots,x_{d} \rangle.
\end{align}
Note that the action of $\mbox{C}^{\otimes d-1}Z$ does not depend on the choice of control qubits. Moreover, $\mbox{C}^{\otimes d-1}Z$ belongs to the $d$-th level of the Clifford hierarchy but is outside the $(d-1)$-th level, which can be seen from the following relations
\begin{eqnarray}
\mbox{K}[R_n,X] &=& e^{-2\pi i/2^n} R_{n-1} \propto R_{n-1},\\
\mbox{K}\left[\mbox{C}^{\otimes n-1}Z , X\otimes I^{\otimes n-1}\right] &=& I\otimes \mbox{C}^{\otimes n-2}Z
\end{eqnarray}
where the commutator is defined as $\mbox{K}[A,B]=ABA^{\dagger}B^{\dagger}$.

We label logical $X$ and $Z$ operators in the color code by $\overline{X^{(i)}}$ and $\overline{Z^{(i)}}$ for $i=1,\ldots,d$. Namely, $\overline{Z^{(i)}}$ is a string-like logical operator parallel to the direction $\hat{i}$ (i.e. connecting two opposite boundaries of color $C_i$) and $\overline{X^{(i)}}$ is a $(d-1)$-dimensional membrane-like logical operator perpendicular to the direction $\hat{i}$. We define the operator $\widetilde{R_{i}}$ recursively for $i=d-1,\ldots,1$ as follows
\begin{eqnarray}
\widetilde{R_{d-1}} &=& \mbox{K}\left[\ \widetilde{R_{d}}, \overline{X^{(1)}}\ \right],\\
\widetilde{R_{d-2}} &=& \mbox{K}\left[\ \widetilde{R_{d-1}}, \overline{X^{(2)}} \right],\\
&\vdots  \\
\widetilde{R_{1}} &=& \mbox{K}\left[\ \widetilde{R_{2}}, \overline{X^{(d-1)}} \right] = \overline{Z^{(d)}}.
\end{eqnarray}
Note that the above relations hold for any permutation of colors $C_{1},\ldots,C_{d}$. Let $\overline{X_j}$ and $\overline{Z_j}$ be logical $X$ and $Z$ operators in the toric code $TC(\mathcal{L}_{j})$. Then, the following correspondence holds
\begin{equation}
\overline{X^{(j)}} \leftrightarrow \overline{X_j}, \qquad \overline{Z^{(j)}} \leftrightarrow \overline{Z_j}.
\end{equation}
We can verify that using $\widetilde{R_{d}}$ one can implement the logical $d$-qubit control-$Z$ gate on the stack of $d$ copies of the toric code. Namely,
\begin{eqnarray}
I \otimes \overline{\mbox{C}^{\otimes d-2}Z} &\propto& \mbox{K}\left[ \overline{\mbox{C}^{\otimes d-1}Z}, \overline{X_1} \right],\\
I^{\otimes 2} \otimes \overline{\mbox{C}^{\otimes d-3}Z} &\propto& \mbox{K}\left[ I \otimes \overline{\mbox{C}^{\otimes d-2}Z}, \overline{X_2} \right],\\
&\vdots&  \\
I^{\otimes d-1} \otimes \overline{Z} &\propto& \mbox{K}\left[ \overline{\mbox{C}Z}, \overline{X_{d-1}} \right].
\end{eqnarray}
In the above equations proportionality indicates the same action of the operators on the code space. Since the disentangling unitary $U$ is a local unitary transformation, $U\widetilde{R_{d}}U^{\dagger}$ is a local unitary transformation implementing the logical $\mbox{C}^{\otimes d-1}Z$ gate on the stack of $d$ copies of the toric code. We summarize the discussion in this section by the following theorem.

\begin{theorem}[Transversal Implementation]
\label{th:transversal}
Consider a $(d+1)$-colorable and $(d+1)$-valent $d$-hypercube-like lattice $\mathcal{L}$ with pairs of boundaries perpendicular to the direction $\hat{i}$ colored in $C_i$ for $i=1,\ldots,d$. Let $\mathcal{L}_{i}$ be a lattice derived from $\mathcal{L}$ by shrinking all $d$-cells of color $C_i$. Then, the logical $d$-qubit control-$Z$ gate can be implemented on $TC(\mathcal{L}_1),\ldots, TC(\mathcal{L}_d)$ --- the stack of $d$-copies of the toric code, by a local unitary transformation
\begin{equation}
\overline{\mbox{C}^{\otimes d-1}Z} \propto U \widetilde{R_d} U^\dag,
\end{equation}
where $\widetilde{R_d}$ is a transversal $R_d$ gate (see Eq.~(\ref{eq:Rtransversal})) implemented in the color code $CC(\mathcal{L})$ and $U$ is the disentangling unitary transforming $CC(\mathcal{L})$ into $\bigotimes_{i=1}^d TC(\mathcal{L}_i)$.
\end{theorem}

Observe that the implementation of $\overline{\mbox{C}^{\otimes d-1}Z}$ seems to require a set of $d$ lattices $\{ \mathcal{L}_{i} \}$ which satisfy certain constraints, i.e. are derived from $\mathcal{L}$ described in the (Transversal Implementation) Theorem~\ref{th:transversal}. In general, it may not be clear whether there exists a local unitary transformation implementing $\overline{\mbox{C}^{\otimes d-1}Z}$ in $d$ copies of the toric code. Yet, one can freely deform the lattices on which the toric code is supported by local operations. Specifically, consider the toric code $TC(\mathcal{L})$ on a $d$-dimensional lattice $\mathcal{L}$. We claim that one can transform $TC(\mathcal{L})$ into $TC(\mathcal{L}')$ by local unitary transformations (and adding or removing ancilla qubits), where $\mathcal{L}'$ is a lattice derived from the original lattice $\mathcal{L}$ by adding or removing edges. Such local deformations of the lattices allow us to obtain $d$ copies of the toric code with $\overline{\mbox{C}^{\otimes d-1}Z}$ implementable by local unitary transformations as long as the boundaries of $d$ copies of the toric code are appropriately arranged. In particular, this implies that three copies of the three-dimensional toric code admit fault-tolerant implementation of a non-Clifford logical gate, which saturates the bound by Bravyi and K\"onig in three dimensions.

\section*{Acknowledgements}

We would like to thank John Preskill, Olivier Landon-Cardinal and Dan Browne for helpful discussions. We acknowledge funding provided by the Institute for Quantum Information and Matter, an NSF Physics Frontiers Center with support of the Gordon and Betty Moore Foundation (Grants No. PHY-0803371 and PHY-1125565). BY is supported by the David and Ellen Lee Postdoctoral fellowship.

\vfill\eject

\section*{Appendix}

Here, we briefly revisit the equivalence of the color code and (multiple decoupled copies of) the toric code in $d$ dimensions without the restriction of point-like excitations. In particular, we focus on the construction of lattices supporting the decoupled copies of the toric code, which can be succinctly described using some notions from algebraic topology. The discussion in the Appendix is presented in the language of the dual lattice unless mention otherwise.

\subsection{Basic definitions of combinatorial geometry}

We start with some basic notions in combinatorial geometry. A $d$-simplex $\delta$ is a convex hull of $d+1$ affinely independent vertices $v_0,v_1,\ldots, v_d$, namely 
\begin{equation}
\delta=\left\{\sum_{i=0}^d t_i v_i \right|\left. 0\leq t_i  \wedge\sum_{i=0}^d t_i = 1\right\}.
\end{equation}
There is a combinatorial definition of a simplex, which we adopt for the rest of the discussion. Namely, a $d$-simplex $\delta$ is the power set of the set of vertices $V=\{v_0,\ldots, v_d \}$ spanning it, $\delta=\mathcal{P}(V)$. A subset $W\subset V$ of size $k+1 \leq d+1$ spans a $k$-simplex $\sigma = \mathcal{P}(W)$, and we call $\sigma$ a $k$-face of $\delta$. We denote the set of all $k$-faces of $\delta$ by $\face{k}{\delta}$. Note that $\face{k}{\delta}=\mathcal{P}_k(V)$, where $\mathcal{P}_k(V)$ denotes the set of subsets of $V$ of cardinality $k$. 

Let $V=\bigsqcup_{i=1}^k W_i$ be a decomposition of the set of vertices $V$ into the union of $k$ disjoint sets $W_1,\ldots,W_k$. Let $\delta=\mathcal{P}(V)$ and $\sigma_i=\mathcal{P}(W_i)$. Then, we can represent $\delta$ as a Cartesian product of its faces $\sigma_1,\ldots,\sigma_k$, namely
\begin{equation}
\delta = \sigma_1 \times \ldots \times \sigma_k.
\end{equation}

We say that $\mathcal{L}$ is a simplicial $d$-complex if it is a set of simplices satisfying the following conditions
\begin{itemize}
\item every face of a simplex in $\mathcal{L}$ is also in $\mathcal{L}$,
\item the intersection of two simplices in $\mathcal{L}$ is a face of both of them,
\item the dimension of the largest simplex in $\mathcal{L}$ is $d$,
\end{itemize}
If in addition
\begin{itemize}
\item for every $k<d$, every $k$-simplex in $\mathcal{L}$ is a face of a $d$-simplex in $\mathcal{L}$,
\end{itemize}
then $\mathcal{L}$ is homogeneous. By $\face k{\mathcal{L}}$ we denote the set of all $k$-simplices belonging to $\mathcal{L}$. An $n$-skeleton of $\mathcal{L}$, denoted by $\mathrm{skel}_n(\mathcal{L})$, is a collection of all $k$-faces of $\mathcal{L}$ for all $k\leq n$, namely $\mathrm{skel}_n(\mathcal{L}) = \bigsqcup_{k=0}^n \face k{\mathcal{L}}$.

We might generalize the notion of a simplex to a cell. Namely, a (closed) $d$-cell $\delta$ is the image of a $d$-dimensional (closed) ball $B^d$ under an attaching map. Similarly to the combinatorial definition of a simplex, we want to think about $\delta$ as a collection of all its $k$-faces, for all $k=0,1,\ldots, d$. We can define a cell complex\footnote{For a rigorous definition of a CW complex, see Ref.~\cite{Hatcher2002}.} in an analogous way to a simplicial complex, allowing for the faces to be cells.

From now on, we only consider complexes containing finitely many simplices (cells). Although a homogeneous simplicial (cell) $d$-complex $\mathcal{L}$ is defined as a collection of simplices (cells), by the same symbol we also denote the union of these simplices (cells) as a topological space. In general, $\mathcal{L}$ is a manifold with a boundary embedded in real space, but for the rest of the discussion we assume $\mathcal{L}$ has no boundary. We also assume $\mathcal{L}$ is a homogeneous simplicial $d$-complex unless stated otherwise.

The $n$-star of $\delta\in\face k{\mathcal{L}}$, denoted by $\star n\delta$, is the set of all $n$-simplices in $\mathcal{L}$ which contain $\delta$ as a face, namely
\begin{equation}
\star n\delta =\{ \sigma\in\face n{\mathcal{L}} \ |\ \sigma\supset\delta  \}.
\end{equation}
Note that $\sigma\in\star n\delta \iff \delta\in\face k\sigma$.

The $n$-link of $\delta\in\face k{\mathcal{L}}$, denoted by $\link n\delta$, is the set of all $n$-simplices in $\mathcal{L}$ which are the $n$-faces of $d$-simplices containing $\delta$, but do not intersect with $\delta$, namely
\begin{equation}
\link n\delta =\{ \sigma\in\face n{\mathcal{L}}\ |\ \sigma\subset\tau\in\star d\delta \wedge \sigma\cap\delta=\emptyset \}.
\end{equation}

Observe that for a $k$-simplex $\delta$ in $\mathcal{L}$ there is a one-to-one mapping between the elements of $\link {{d-k-1}}\delta$ and $\star d\delta$, namely
\begin{equation}
\sigma\in\link{{d-k-1}}\delta \xleftrightarrow{\delta\times\sigma = \tau} \tau\in\star d\delta .
\end{equation}

We say that $\mathcal{L}$ is $(d+1)$-colorable if there exists a function 
\begin{equation}
\textrm{color}:\Delta_0 (\mathcal{L})\rightarrow\mathbb{Z}_{d+1},
\end{equation}
where $\mathbb{Z}_{d+1}=\{0,1,\ldots, d \}$ is the set of $d+1$ colors, and two vertices connected by an edge have different colors. We define $\col{\delta}$ to be the set of colors assigned to vertices of a simplex $\delta$, namely
\begin{equation}
\col{\delta}=\bigsqcup_{v\in\Delta_0 (\delta)}\col{v}.
\end{equation}

Now, we are ready to define the color code and the toric code in $d$ dimensions. The color code is a stabilizer code with the stabilizer group $CC_k(\mathcal{L})$ defined on a $(d+1)$-colorable homogeneous simplicial $d$-complex $\mathcal{L}$, where $k\in\{0,\ldots,d-2\}$. One qubit is placed at each and every $d$-simplex in $\mathcal{L}$, and $X$- and $Z$-type stabilizer generators are associated with $(d-k-2)$- and $k$-simplices as follows
\begin{eqnarray}
\forall \delta\in\face {{d-k-2}}{\mathcal{L}}: X(\delta)=\bigotimes_{\sigma\in \star d\delta} X(\sigma),\\
\forall \delta\in\face {k}{\mathcal{L}}: Z(\delta)=\bigotimes_{\sigma\in \star d\delta} Z(\sigma),
\end{eqnarray}
where $X(\sigma)$ is the Pauli $X$ operator on a qubit placed at $\sigma$; similarly $Z(\sigma)$.

The toric code is a stabilizer code with the stabilizer group $TC_k(\mathcal{L})$ defined on a homogeneous cell $d$-complex $\mathcal{L}$, where $k\in\{1,\ldots,d-1\}$. One qubit is placed at each and every $k$-cell in $\mathcal{L}$, and $X$- and $Z$-type stabilizer generators are associated with $(k+1)$- and $(k-1)$-cells in the following way
\begin{eqnarray}
\forall \delta\in\face {{k+1}}{\mathcal{L}}: X(\delta)=\bigotimes_{\sigma\in \face k\delta} X(\sigma),\\
\forall \delta\in\face {{k-1}}{\mathcal{L}}: Z(\delta)=\bigotimes_{\sigma\in \star k\delta} Z(\sigma).
\end{eqnarray}

\subsection{Equivalence revisited}

Let us revisit the (Equivalence) Theorem~\ref{th:equivalence}. Note that for the sake of simplicity we assumed earlier that the color code on a $d$-dimensional lattice $\mathcal{L}$ has point-like excitations, which corresponds to the $CC_{d-2}(\mathcal{L})$ case. Now we state the equivalence between the color code and the toric code in full generality in the following theorem.

\begin{theorem}
\label{th:general}
Let the topological color code $CC_k ({\mathcal{L}})$ be defined on a $(d+1)$-colorable homogeneous simplicial $d$-complex $\mathcal{L}$ without boundary, where $0\leq k\leq d-2$. Then, there exists a local Clifford unitary $U$ such that
\begin{equation}
U\left[ CC_k (\mathcal{L}) \otimes \mathcal{S}_1\right ]U^\dag = \bigotimes_{N\in\mathcal{P}_{d-1-k}(\mathbb{Z}_{d})} TC_{k+1}(\mathcal{L}_N) \otimes \mathcal{S}_2
\end{equation}
where $\mathcal{S}_1$ and $\mathcal{S}_2$ represent the stabilizer groups of decoupled ancilla qubits, and $TC_{k+1}(\mathcal{L}_N)$ is a copy of the toric code defined on a homogeneous cell $(k+2)$-complex $\mathcal{L}_N$ obtained from $\mathcal{L}$ by removing all simplices with faces of colors in $N$. Moreover, one can choose the disentangling unitary $U$ to be of the form
\begin{equation}
\label{eq:appendix}
U=\bigotimes_{\substack{\delta\in\face 0{\mathcal{L}}\\ \col{\delta} = \{ d \} }}  U(\delta),
\end{equation}
where $U(\delta)$ is a Clifford unitary acting only on qubits placed on $d$-simplices in $\star d{\delta}$, and some ancilla qubits associated with $\delta$.
\end{theorem}

Note that after the disentangling one obtains ${d\choose d-1-k}={d\choose k+1}$ decoupled copies of the toric code, enumerated by different choices of the subset $N$ of $d-1-k$ colors from $\mathbb{Z}_d$. Moreover, one might need to locally add ancilla qubits either to the color code, or the toric code depending on the simplicial $d$-complex $\mathcal{L}$. Clearly, the (Equivalence) Theorem~\ref{th:equivalence} is a special case of the Theorem~\ref{th:general}, with $k=d-2$, $\mathcal{S}_2 = \emptyset$ and $C_0 = d$. The rest of the Appendix is devoted to the construction of the cell complexes supporting decoupled copies of the toric code and the explanation of how to find a local Clifford unitary $U$.

To obtain $\mathcal{L}_N$ from $\mathcal{L}$, where $N\in\mathcal{P}_{d-1-k}(\mathbb{Z}_d)$, one follows the following procedure.
\begin{enumerate}
\item  Take the $(k+1)$-skeleton $\mathrm{skel}_{k+1}(\mathcal{L})$ of $\mathcal{L}$ and construct a new $(k+1)$-skeleton, $\text{skel}'_{k+1}(\mathcal{L})$, by removing from $\mathrm{skel}_{k+1}(\mathcal{L})$ all simplices with faces of colors in $N$, namely
\begin{equation}
\text{skel}'_{k+1}(\mathcal{L}) = \{\sigma\in\text{skel}_{k+1}(\mathcal{L}) | \col{\sigma}\subset \mathbb{Z}_{d+1}\setminus N \}.
\end{equation}
\item For every $\tau\in\face{{d-2-k}}{\mathcal{L}}$, such that $\col{\tau}= N$, attach a $(k+2)$-cell to $\link{{k+1}}{\tau}\subset\text{skel}'_{k+1}(\mathcal{L})$. Resulting $(k+2)$-skeleton is $\mathcal{L}_N$.
\end{enumerate}
Note that in Step 2 we used a fact that $\link {{k+1}}{\tau}$ is homomorphic to a $(k+1)$-sphere, and thus we can attach a $(k+2)$-ball to $\link {{k+1}}{\tau}$ (see Ref.~\cite{Glaser_Text} for a proof and an illustrative discussion on combinatorial manifolds).

The disentangling unitary $U$ in Eq.~(\ref{eq:appendix}) has a tensor product structure. Thus, let us have a closer look at one of its constituents, $U(\delta)$, where $\delta$ is a $0$-simplex in $\mathcal{L}$ of color $\{d\}$. Let $U(\delta)$ to be a Clifford unitary transforming the Hilbert space of qubits placed on $d$-simplices in $\star d{\delta}$ and $A = |\star{{k+1}}{\delta}| - |\star d{\delta}|$ ancilla qubits\footnote{If $A<0$, then $U(\delta)$ is a map between $\mathcal{H}(\star d{\delta})$ and $\mathcal{H}(\star{{k+1}}{\delta})\otimes \mathcal{H}_{ancilla}$. }, $\mathcal{H}(\star d{\delta})\otimes \mathcal{H}_{ancilla}$, into the Hilbert space $\mathcal{H}(\star{{k+1}}{\delta})$ of qubits placed on $(k+1)$-simplices in $\star {{k+1}}{\delta}$. Let $\mathcal{O}_{CC}$ be the overlap group of the color code $CC_k(\mathcal{L})$ on the set of qubits $\star d{\delta}$, and $\mathcal{S}$ be the stabilizer group generated by single qubit Pauli $Z$ operators on the ancilla qubits, namely
\begin{eqnarray}
\mathcal{O}_{CC} &=& \left\langle \bigotimes_{\alpha\in\star d{\sigma}} X(\alpha), \bigotimes_{\alpha\in\star d{\tau}} Z(\alpha) \Bigg\vert\ \forall \sigma\in\star{{d-1-k}}{\delta}, \tau\in\star{{k+1}}{\delta} \right\rangle,\\
\mathcal{S} &=& \left\langle Z_i \vert\ \forall i\in\{1,\ldots,A\} \right\rangle.
\end{eqnarray}
Let $\mathcal{O}_{TC}$ a group of operators acting on qubits placed on $\star{{k+1}}{\delta}$ defined as follows
\begin{equation}
\mathcal{O}_{TC} = \left\langle \bigotimes_{\alpha\in \link{{k+1}}{\sigma} \cap \star{{k+1}}{\delta}} X(\alpha), Z(\tau) \Bigg\vert\ \forall \sigma\in\link{{d-2-k}}{\delta}, \tau\in\star{{k+1}}{\delta} \right\rangle,
\end{equation}
We require that $U(\delta)$ maps the generators of $\mathcal{O}_{CC}\otimes\mathcal{S}$ into the generators of $\mathcal{O}_{TC}$ according to the following rules
\begin{eqnarray}
\forall \sigma\in\link{{d-2-k}}{\delta}:\quad g_i =& \left( \bigotimes_{\alpha\in\star d{\sigma\times\delta}} X(\alpha) \right) &\rightarrow\quad h_i = \bigotimes_{\alpha\in \link{{k+1}}{\sigma} \cap \star{{k+1}}{\delta}} X(\alpha),\\
\forall \tau\in\star{{k+1}}{\delta}:\quad g'_i =& \left( \bigotimes_{\alpha\in\star d{\tau}} Z(\alpha) \right) &\rightarrow\quad h'_i = Z(\tau),\\
 i\in\{1,\ldots,A\}:\quad g''_i =& Z_i &\rightarrow\quad h''_i\in Z(\mathcal{O}_{TC}).
\end{eqnarray}
where the parenthesis indicate that the mapping holds up to multiplication by elements of the center $Z(\mathcal{O}_{CC}\otimes\mathcal{S})$ and we choose $\{ h_i,h'_i,h''_i \}$ to be independent. Note that we had to add $A= |\star{{k+1}}{\delta}| - |\star d{\delta}|$ ancilla qubits to guarantee that $\mathcal{O}_{CC}\otimes\mathcal{S},\mathcal{O}_{TC}\in\mathbf{Pauli}(n= |\star{{k+1}}{\delta}|)$.  One can check that $\{g_i,g'_i,g''_i \}$ and $\{h_i,h'_i,h''_i \}$ have the same commutation relations. The existence of the unitary $U(\delta)$ follows from the (Clifford Transformation) Lemma~\ref{lemma:Cliffordunitary}, given $\mathcal{O}_{CC}\otimes\mathcal{S}$ and $\mathcal{O}_{TC}$ are isomorphic. We will show this fact invoking the (Isomorphic Groups) Lemma~\ref{lemma:isomorphicgroups}.

We want to verify that $G(\mathcal{O}_{CC}\otimes\mathcal{S}) = G(\mathcal{O}_{TC})$ and $G(Z(\mathcal{O}_{CC}\otimes\mathcal{S})) = G(Z(\mathcal{O}_{TC}))$. First note that the elements of the center $Z(\mathcal{O}_{TC})$ are generated by only $Z$-type operators which are derived from $k$-cells, namely
\begin{equation}
Z(\mathcal{O}_{TC}) = \left\langle \bigotimes_{\substack{\tau\in\star{{k+1}}{\sigma}\\\col{\tau}=\col{\sigma}\sqcup \{ n_1\}}}
Z(\tau)\Bigg\vert\ \forall \sigma\in\star k{\delta}, n_1 \in\mathbb{Z}_d\setminus\col{\sigma}\right\rangle.
\end{equation}
Since for any $\sigma\in\star k{\delta}$ we can choose $n_1 \in\mathbb{Z}_d\setminus\col{\sigma}$ in ${d-k \choose 1}$ ways, then there are ${d-k\choose 1} |\star k{\delta}|$ generators of $Z(\mathcal{O}_{TC})$. Note that not all of them are independent. Rather, they have to satisfy certain relations, which we call $\mathcal{R}_1$, namely
\begin{eqnarray}
\forall \rho_1 \in\star {{k-1}}{\delta}, \{n_1,n_2\}\subset \mathbb{Z}_d\setminus\col{\rho_1}:\\
\prod_{\substack{\sigma\in\star k{\rho_1}\\ \col{\sigma} \subset \col{\rho_1}\sqcup \{ n_1,n_2\}}} 
\left(\bigotimes_{\substack{\tau\in\star{{k+1}}{\sigma}\\\col{\tau}=\col{\rho_1}\sqcup \{ n_1,n_2\}}}
 Z(\tau)\right)= I.
\end{eqnarray}
Since for any $\rho_1 \in\star {{k-1}}{\delta}$ we can choose $\{n_1,n_2\} \subset\mathbb{Z}_d\setminus\col{\rho_1}$ in ${d-k+1 \choose 2}$ ways, then there are ${d-k+1\choose 2} |\star {{k-1}}{\delta}|$ relations $\mathcal{R}_1$. Note that not all relations $\mathcal{R}_1$ are independent. They have to satisfy ${d-k+2\choose 3} |\star {{k-2}}{\delta}|$ relations $\mathcal{R}_2$ obtained for any choice of $\rho_2\in\star {{k-2}}{\delta}$ and $\{n_1,n_2,n_3\} \subset\mathbb{Z}_d\setminus\col{\rho_2}$. But relations $\mathcal{R}_2$ are not independent, and so on. Proper counting of independent relations between generators of $Z(\mathcal{O}_{TC})$ gives the following alternating sum $|\mathcal{R}_1|-|\mathcal{R}_2|+|\mathcal{R}_3|-\ldots +(-1)^{k-1}|\mathcal{R}_k|$. Once the constraints have been properly accounted for, since $G(Z(\mathcal{O}_{TC}))$ is equal to the number of generators minus the number of independent relations between them, then we obtain
\begin{equation}
G(Z(\mathcal{O}_{TC})) = {d-k\choose 1} |\star k{\delta}| - {d -k +1\choose 2}  |\star {{k-1}}{\delta}| + \ldots + (-1)^k {d\choose k+1}|\star 0{\delta}|.
\label{eq:gensofcentertoric}
\end{equation}
Using the fact that the toric code on an $n$-sphere does not encode logical qubits, we obtain that the number of independent $X$-type generators of $\mathcal{O}_{TC}$ is equal to $|\star{{k+1}}{\delta}| - G(Z(\mathcal{O}_{TC}))$. Thus, including $|\star{{k+1}}{\delta}|$ independent $Z$-type generators,
\begin{equation}
G(\mathcal{O}_{TC}) = 2|\star{{k+1}}{\delta}| - G(Z(\mathcal{O}_{TC})).
\label{eq:gensoftoric}
\end{equation}

To analyze the number of independent generators of $\mathcal{O}_{CC}\otimes\mathcal{S}$ and its center $Z(\mathcal{O}_{CC}\otimes\mathcal{S})$, we use results from the Appendix~D in Ref.~\cite{Bombin2007}. First, let us rephrase our problem in the language of the primal lattice. Note that a $0$-simplex $\delta$ corresponds to a $d$-cell $c$, qubits are placed on vertices of $c$, and $X$- and $Z$-type stabilizers are supported on qubits on  vertices of $(k+2)$- and $(d-k)$-faces of $c$. Let $\partial c$ be the boundary of $c$, which can be thought of as a $d$-colorable and $d$-valent homogeneous cell $(d-1)$-complex. Let us denote by $C_i$ the number of $i$-faces of $c$, where $i=0,\ldots,d$. Clearly, $C_i = |\star{{d-i}}{\delta}|$. The overlap group of the color code on the qubits of $c$ is thus generated by $X$- and $Z$-type operators on $(k+1)$- and $(d-k-1)$-faces of $c$. Note that $(k+1)$- and $(d-k-1)$-faces of $c$ can be thought of as faces of $\partial c$, and thus number of independent generators of $\mathcal{O}_{CC}\otimes\mathcal{S}$  is
\begin{equation}
G(\mathcal{O}_{CC}\otimes\mathcal{S}) = C_{k+1} - I(d-1,k+1) + C_{d-k-1} -I(d-1,d-k-1) + C_{d-k-1} - C_0,
\end{equation}
where we include $A=C_{d-k-1} - C_0$ single qubit Pauli $Z$ operators on the ancilla qubits. By $I(d-1,i)$ we denote the number of independent relations between operators on $i$-cells of $\partial c$. In particular (see Eq.~(D14) in Ref.~\cite{Bombin2007}),
\begin{eqnarray}
\label{eq:definitionI}
I(d-1,s) &=& {d-1 \choose s-1} \sum_{i=0}^{d-1-s} (-1)^i h_{s+i} + \sum_{i=0}^{d-s-2} {s+i \choose s-1}(-1)^i C_{s+i+1}\\
&=& {d-1 \choose s-1} (-1)^{d-1-s}+ \sum_{i=0}^{d-s-2} {s+i \choose i+1}(-1)^i C_{s+i+1},
\end{eqnarray}
where $h_i$ is the $i$-th Betti number of $\partial c$. Since $\partial c$ is homomorphic to a $(d-1)$-sphere, then $h_i=1$ if $i=0,d-1$; otherwise $h_i=0$. The center $Z(\mathcal{O}_{CC})$ is generated by $X$- and $Z$-type operators on $(k+2)$- and $(d-k)$-faces of $c$, and single-qubit Pauli $Z$ operators on ancilla qubits. Thus\footnote{If $k=0,d-2$, then we set $I(d-1,d)=0$.},
\begin{equation}
G(Z(\mathcal{O}_{CC}\otimes\mathcal{S})) = C_{k+2} - I(d-1,k+2) + C_{d-k} -I(d-1,d-k) + C_{d-k-1} - C_0.
\end{equation}
We can express Eqs.~(\ref{eq:gensofcentertoric})~and~(\ref{eq:gensoftoric}) is terms of $C_i$'s, namely
\begin{eqnarray}
G(Z(\mathcal{O}_{TC})) &=& {d-k\choose 1} C_{d-k} - {d-k +1\choose 2}  C_{d-k+1} + \ldots + (-1)^k {d\choose k+1} C_d\\
&=& \sum_{i=0}^k (-1)^i {d-k+i\choose i+1} C_{d-k+i},\\
G(\mathcal{O}_{TC}) &=& 2C_{d-k-1}-G(Z(\mathcal{O}_{TC})).
\end{eqnarray}
There are many relations between the number of $i$-cells of $\partial c$, which is a $d$-colorable and $d$-valent homogeneous cell $(d-1)$-complex homomorphic to a $(d-1)$-sphere. In particular, the following identities hold
\begin{equation}
\label{eq:conditionsonC}
-{d-1 \choose s}\chi + (-1)^s C_0 + \sum_{i=0}^{s-1} (-1)^i {d-2-i \choose d-1-s} C_{d-1-i} + \sum_{i=s+1}^{d-1} (-1)^i {i-1 \choose s} C_i = 0,
\end{equation}
for any $s=0,\ldots,d-1$, where $\chi = 1+(-1)^{d-1}$ (see Eq.~(D16) in Ref.~\cite{Bombin2007}). One can straightforwardly verify that $G(\mathcal{O}_{CC}\otimes\mathcal{S}) - G(\mathcal{O}_{TC})=0$ and $G(Z(\mathcal{O}_{CC}\otimes\mathcal{S})) - G(Z(\mathcal{O}_{TC})) = 0$, since they are obtained from Eq.~(\ref{eq:conditionsonC}) by setting $s=k$ and $s=k+1$, respectively. This finishes the proof that $\mathcal{O}_{CC}\otimes\mathcal{S}$ and $\mathcal{O}_{TC}$ are isomorphic.

\end{document}